\documentclass[runningheads]{llncs}
\usepackage{graphicx}
\usepackage{amsmath}
\usepackage{stmaryrd}
\usepackage{booktabs}   
\usepackage[dvipsnames]{xcolor}
\usepackage[colorlinks=true,citecolor=ForestGreen,linkcolor=red]{hyperref}
\usepackage{cleveref}
\usepackage{proof-dashed}
\usepackage{tikz-cd}

\usepackage{latexsym}
\usepackage{amssymb}            
\usepackage{stmaryrd}           
\usepackage{color}

\newcommand{\m}[1]{\mathsf{#1}}
\newcommand{\mb}[1]{\mathbf{#1}}

\newcommand{\mi}[1]{\mbox{\it #1}}

\newcommand{\CC}{\mathcal{C}}

\newcommand{\up}{{\uparrow}}
\newcommand{\down}{{\downarrow}}
\newcommand{\uup}[2]{{{}^{#1}_{#2}}{{\uparrow}\kern-0.2em{\downarrow}}}
\newcommand{\ddown}[2]{{{}^{#1}_{#2}}{{\downarrow}\kern-0.2em{\uparrow}}}
\newcommand{\uscore}{\mbox{\tt\char`\_}}

\newcommand{\mmode}[1]{{\mathchoice{\m{#1}}{\m{#1}}{\scriptscriptstyle\m{#1}}{\scriptscriptstyle\m{#1}}}}

\newcommand{\mS}{\mmode{S}}

\newcommand{\mN}{\mmode{N}}

\newcommand{\seq}{\vdash}

\newcommand{\semi}{\mathrel{;}}

\newcommand{\cut}{\m{cut}}

\newcommand{\id}{\m{id}}

\newcommand{\vvdash}{\mathrel{\vdash\kern-0.8ex\vdash}}

\newcommand{\lolli}{\multimap}
\newcommand{\tensor}{\otimes}
\newcommand{\with}{\mathbin{\binampersand}}

\newcommand{\one}{\mathbf{1}}

\newcommand{\bang}{{!}}

\newcommand{\eval}{\m{eval}}

\newcommand{\cont}{\m{cont}}

\newenvironment{rules}{\[\begin{array}{c}}{\end{array}\]}
\newenvironment{srules}{\[\begin{small}\begin{array}{c}}{\end{array}\end{small}\]}
\newenvironment{crules}[1]{\[\hspace*{#1}\begin{small}\begin{array}{c}}{\end{array}\end{small}\]}
\newcommand{\SAX}{\textsc{Sax}}
\newcommand{\langname}{{Seax}}
\newcommand{\larrow}{\;}

\begin{document}
\title{Back to Futures}

\author{Klaas Pruiksma \and Frank Pfenning}
\institute{Carnegie Mellon University \\
\email{kpruiksm@andrew.cmu.edu} \qquad \email{fp@cs.cmu.edu}}

%
%
\authorrunning{K. Pruiksma and F. Pfenning}
%
%
\maketitle              
\begin{abstract}
Common approaches to concurrent programming begin with languages whose
semantics are naturally sequential and add new constructs that provide
limited access to concurrency, as exemplified by \emph{futures}.
This approach has been quite successful, but often does not provide a
satisfactory theoretical backing for the concurrency constructs, and
it can be difficult to give a good semantics that allows a programmer to use
more than one of these constructs at a time.

We take a different approach, starting with a concurrent language based
on a Curry-Howard interpretation of adjoint logic, to which we add three
atomic primitives that allow us to encode sequential composition and various
forms of synchronization. The resulting language is highly expressive,
allowing us to encode futures, fork/join parallelism, and monadic concurrency
in the same framework. Notably, since our language is based on adjoint logic, we
are able to give a formal account of \emph{linear futures}, which have been used
in complexity analysis by Blelloch and Reid-Miller. The uniformity of this approach
means that we can similarly work with many of the other concurrency primitives in
a linear fashion, and that we can  mix several of these forms of concurrency in the
same program to serve different purposes.

\keywords{Shared-memory concurrency  \and Adjoint Logic \and Futures}
\end{abstract}
\section{Introduction}


Concurrency has been a very useful tool in increasing performance
of computations and in enabling distributed computation, and consequently,
there are a wide variety of different approaches to programming languages
for concurrency. A common pattern is to begin with a sequential language
and add some form of concurrency primitive, ranging from threading libraries
such as pthreads to monads to encapsulate concurrent computation,
as in SILL~\cite{Toninho13esop,Toninho15phd,Griffith16phd}, to
futures~\cite{Halstead85}. Many of these approaches have seen great practical
success, and yet from a theoretical perspective, they are often unsatisfying,
with the concurrent portion of the language being attached to the sequential
base language in a somewhat ad hoc manner, rather than having a coherent
theoretical backing for the language as a whole.

In order to give a more uniform approach to concurrency, we take the opposite
approach and begin with a language, \langname, whose semantics are naturally
concurrent. With a minor addition to \langname, we are able to force synchronization,
allowing us to encode sequentiality. In the resulting language, we can
model many different concurrency primitives, including futures, fork/join,
and concurrency monads. Moreover, as all of these constructs are encoded in
the same language, we can freely work with any mixture and retain the same
underlying semantics and theoretical underpinnings.

Two lines of prior research meet in the development of \langname. The first
involves a new presentation of intuitionistic logic, called the
\emph{semi-axiomatic sequent calculus} (\SAX)~\cite{DeYoung20fscd},
which combines features from Hilbert's axiomatic form~\cite{Hilbert34}
and Gentzen's sequent calculus~\cite{Gentzen35}. Cut reduction in the
semi-axiomatic sequent calculus can be put into correspondence with
asynchronous communication, either via message
passing~\cite{Pruiksma19places} or via shared
memory~\cite{DeYoung20fscd}. We build on the latter, extending it in
three major ways to get \langname. First, we extend from intuitionistic
logic to a semi-axiomatic presentation of adjoint
logic~\cite{Reed09un,Licata16lfcs,Licata17fscd,Pruiksma19places}, the second
major line of research leading to \langname. This gives  us a richer set of
connectives as well as the ability to work with linear and other substructural
types. Second, we add equirecursive types and recursively defined processes,
allowing for a broader range of programs, at the expense of termination, as usual.
Third, we add three new \emph{atomic write} constructs that write a value and its
tag in one step. This minor addition enables us to encode both some forms
of synchronization and sequential composition of processes, despite the naturally
concurrent semantics of \langname.

This resulting language is highly expressive. Using these features, we are able to
model functional programming with a semantics in destination-passing style
that makes memory explicit~\cite{Wadler84slfp,Larus89thesis,Cervesato02tr,Simmons12phd},
allowing us to write programs in more familiar functional syntax which can then
be expanded into \langname. We can also encode various forms of concurrency
primitives, such as fork/join parallelism~\cite{Conway63afips} implemented by
parallel pairs, futures~\cite{Halstead85}, and a concurrency monad in the style
of SILL~\cite{Toninho13esop,Toninho15phd,Griffith16phd}
(which combines sequential functional with concurrent session-typed programming).
As an almost immediate consequence of our reconstruction of futures,
we obtain a clean and principled subsystem of \emph{linear
  futures}, already anticipated and used in parallel complexity
analysis by Blelloch and Reid-Miller~\cite{Blelloch99tocs} without
being rigorously developed.

The principal contributions of this paper are:
\begin{enumerate}
\item the language \langname, which has a concurrent write-once
  shared-memory semantics for programs based on a
  computational interpretation of adjoint logic;
\item a model of sequential computation using an extension of this
  semantics with limited atomic writes;
\item a reconstruction of fork/join parallelism;
\item a reconstruction of futures, including a rigorous
  definition of linear futures;
\item a reconstruction of a concurrency monad which combines functional
  programming with session-typed concurrency as an instance of the
  adjoint framework;
\item the uniform nature of these reconstructions, which allows us
  to work with any of these concurrency primitives and more all within
  the same language.
\end{enumerate}

We begin by introducing the type system and syntax for \langname,
along with some background on adjoint logic (\cref{sec:type-system}),
followed by its semantics, which are naturally concurrent
(\cref{sec:sem-concur}). At this point, we are able to look at some
examples of programs in \langname. Next, we make the critical addition
of sequentiality (\cref{sec:sem-seq}), examining both what changes we
need to make to \langname\ to encode sequentiality and how we go about
that encoding. Using our encoding of sequentiality, we can build a
reconstruction of a standard functional language's lambda terms
(\cref{sec:functions}), which both serves as a simple example of a
reconstruction and illustrates that we need not restrict ourselves
to the relatively low-level syntax of {\langname} when writing programs.
Following this, we examine and reconstruct several concurrency primitives,
beginning with futures (\cref{sec:futures}), before moving on to parallel
pairs (an implementation of fork/join, in \cref{sec:fork-join}) and a
concurrency monad that borrows heavily from SILL (\cref{sec:sill}). We
conclude with a brief discussion of our results and future work.

\section{\langname: Types and Syntax}
\label{sec:type-system}

The type system and language we present here, which we will use throughout this paper,
are based on adjoint logic~\cite{Reed09un,Licata16lfcs,Licata17fscd,Pruiksma19places,Pruiksma20arxiv}
starting with a Curry-Howard interpretation, which we then leave behind by adding recursion, allowing
a richer collection of programs. Most of the details of adjoint logic are not relevant here,
and so we provide a brief overview of those that are, focusing on how they relate to our language.

In adjoint logic, propositions are stratified into distinct layers,
each identified by a \emph{mode}. For each mode $m$ there is a set
$\sigma(m) \subseteq \{W,C\}$ of structural properties satisfied by
antecedents of mode $m$ in a sequent. Here, $W$ stands for weakening
and $C$ for contraction. For simplicity, we always assume exchange is
possible. In addition, any instance of adjoint logic specifies a
preorder $m \geq k$ between modes, expressing that the proof of
a proposition $A_k$ of mode $k$ may depend on assumptions $B_m$. In
order for cut elimination (which forms a basis for our semantics)
to hold, this ordering must be compatible with the structural properties:
if $m \geq k$ then $\sigma(m) \supseteq \sigma(k)$.  Sequents then have
the form $\Gamma \vdash A_k$ where, critically, each antecedent $B_m$ in
$\Gamma$ satisfies $m \geq k$.  We express this concisely as
$\Gamma \geq k$.

We can go back and forth between the layers using \emph{shifts}
$\up_k^m A_k$ (up from $k$ to $m$ requiring $m \geq k$) and
$\down^r_m A_r$ (down from $r$ to $m$ requiring $r \geq m$). A given pair
$\up_k^m$ and $\down^m_k$ forms an adjunction, justifying the name
adjoint logic.

\begin{figure}[!htb]
\[
  \begin{array}{llcll}
    \mbox{Types} & A_m, B_m & ::= & \oplus_m\{\ell : A^\ell_m\}_{\ell \in L} & \mbox{internal choice} \\
                 & & \mid & A_m \tensor_m B_m & \mbox{multiplicative conjunction} \\
                 & & \mid & \one_m & \mbox{multiplicative unit} \\
                 & & \mid & \with_m\{\ell : A^\ell_m\}_{\ell \in L} & \mbox{external choice} \\
                 & & \mid & A_m \lolli_m B_m & \mbox{linear implication} \\ 
                 & & \mid & \down^r_m A_r & \mbox{down shift, $r \geq m$} \\
                 & & \mid & \up_k^m A_k & \mbox{up shift, $m \geq k$} \\
                 & & \mid & t & \mbox{type variables} \\[1em]
    \mbox{Processes} &
    P, Q & ::= & x_m \leftarrow P \semi Q & \mbox{cut: allocate $x$, spawn $P$, continue as $Q$} \\
                 & & \mid & x_m \leftarrow y_m & \mbox{id: copy or move from $y$ to $x$} \\
                 & & \mid & x_m.V & \mbox{write $V$ to $x$, or read $K$ from $x$ and pass it $V$} \\
                 & & \mid & \mb{case}\, x_m\, K & \mbox{write $K$ to $x$, or read $V$ from $x$ and pass it to $K$} \\
                 & & \mid & x_k \leftarrow p \leftarrow \overline{y_m} & \mbox{call $p$ with $\overline{y}$ binding $x$} \\[1em]
    \mbox{Values} &
    V & ::= & i(y_m) & \mbox{label $i$ tagging address $y$ ($\oplus$, $\with$)} \\
                 & & \mid & \langle w_m, y_m\rangle & \mbox{pair of addresses $w$ and $y$ ($\tensor$, $\lolli$)} \\
                 & & \mid & \langle\,\rangle & \mbox{unit value ($\one$)} \\
                 & & \mid & \mb{shift}(y_k) & \mbox{shifted address $y_k$ ($\down$, $\up$)} \\[1em]
    \mbox{Continuations} &
    K & ::= & (\ell(y_m) \Rightarrow P_\ell)_{\ell \in L} & \mbox{branch on $\ell \in L$ to bind $y$ ($\oplus$, $\with$)} \\
                 & & \mid & (\langle w_m, y_m\rangle \Rightarrow P) & \mbox{match against pair to bind $\langle w, y\rangle$ ($\tensor$, $\lolli$)} \\
                 & & \mid & (\langle\,\rangle \Rightarrow P) & \mbox{match against unit element ($\one$)} \\
                 & & \mid & (\mb{shift}(y_k) \Rightarrow P) & \mbox{match against $\mb{shift}$ to bind $y_k$ ($\down$, $\up$)}

  \end{array}
\]
\caption{Types and process expressions}
\label{fig:types-and-procs}
\end{figure}

\begin{figure}[!htb]
\begin{crules}{-5em}
  \infer[\m{cut}]
  {\Gamma_C, \Delta, \Delta' \vdash (x \leftarrow P \semi Q) :: (z : C_r)}
  {(\Gamma_C, \Delta \geq m \geq r)
    & \Gamma_C, \Delta \vdash P :: (x : A_m)
    & \Gamma_C, \Delta', x : A_m \vdash Q :: (z : C_r)}
  \qquad
  \infer[\m{id}]
  {\Gamma_W, y : A_m \vdash x \leftarrow y :: (x : A_m)}
  {\mathstrut}
  \\[1em]
  \infer[\m{call}]
  {\Gamma \vdash z \leftarrow p \leftarrow \overline{w} :: (z : A_k)}
  {\overline{B_m} \vdash p :: A_k \in \Sigma & \Gamma \vdash \overline{w : B_m}}
  \quad
  \infer[\m{call\_var}_\alpha]
  {\Gamma, x : B_m \vdash (\Delta, x : B_m)}
  {\Gamma, (x : B_m)^\alpha \vdash \Delta}
  \quad
  \infer[\m{call\_empty}]
  {\Gamma_W \vdash (\cdot)}
  {\mathstrut}
  \\[1em]
  \infer[{\oplus}R^0]
  {\Gamma_W, y : A_m^i \vdash x.i(y) :: (x : \oplus_m\{\ell : A^\ell_m\}_{\ell \in L})}
  {(i \in L)}
  \\[1em]
  \infer[{\oplus}L_\alpha]
  {\Gamma, x : \oplus_m\{\ell : A_m^\ell\}_{\ell \in L} \vdash \mb{case}\, x\, (\ell(y) \Rightarrow Q_\ell)_{\ell \in L}
    :: (z : C_r)}
  {\Gamma, (x : \oplus_m\{\ell : A_m^\ell\}_{\ell \in L})^\alpha,
    y : A_m^\ell \vdash Q_\ell :: (z : C_r) & \mbox{(for all $\ell \in L$)}}
  \\[1em]
  \infer[{\with}R]
  {\Gamma \vdash \mb{case}\, x\, (\ell(y) \Rightarrow P_\ell)_{\ell \in L} :: (x : {\with}_m\{\ell : A_m^\ell\}_{\ell \in L})}
  {\Gamma \vdash P_\ell :: (y : A_m^\ell) & \mbox{(for all $\ell \in L$)}}
  \qquad
  \infer[{\with}L^0]
  {\Gamma_W, x : {\with}_m\{\ell : A_m^\ell\}_{\ell \in L} \vdash x.i(y) :: (y : A_m^i)}
  {(i \in L)}
  \\[1em]
  \infer[{\one}R^0]
  {\cdot \vdash x.\langle\, \rangle :: (x : \one_m)}
  {\mathstrut}
  \qquad
  \infer[{\one}L_\alpha]
  {\Gamma, x : \one_m \vdash \mb{case}\, x\, (\langle\, \rangle \Rightarrow P) :: (z : C_r)}
  {\Gamma, (x : \one_m)^\alpha \vdash P :: (z : C_r)}
  \\[1em]
  \infer[{\lolli}R]
  {\Gamma \vdash \mb{case}\, x\, (\langle w,y\rangle \Rightarrow P) :: (x : A_m \lolli_m B_m)}
  {\Gamma, w : A_m \vdash P :: (y : B_m)}
  \qquad
  \infer[{\lolli}L^0]
  {\Gamma_W, w : A_m, x : A_m \lolli_m B_m \vdash x.\langle w,y\rangle :: (y : B_m)}
  {\mathstrut}
  \\[1em]
  \infer[{\tensor}R^0]
  {\Gamma_W, w : A_m, y : B_m \vdash x.\langle w,y\rangle :: (x : A_m \tensor_m B_m)}
  {\mathstrut}
  \qquad
  \infer[{\tensor}L_\alpha]
  {\Gamma, x : A_m \tensor_m B_m \vdash \mb{case}\, x\, (\langle w,y\rangle \Rightarrow P) :: (z : C_r)}
  {\Gamma, (x : A_m \tensor_m B_m)^\alpha, w : A_m, y : B_m \vdash P :: (z : C_r)}
  \\[1em]
  \infer[{\down}R^0]
  {\Gamma_W, y : A_m \vdash x_k.\mb{shift}(y_m) :: (x : \down^m_k A_m)}
  {\mathstrut}
  \quad
  \infer[{\down}L_\alpha]
  {\Gamma, x : \down^m_k A_m \vdash \mb{case}\, x_k\, (\mb{shift}(y_m) \Rightarrow Q) :: (z :: C_r)}
  {\Gamma, (x : \down^m_k A_m)^\alpha, y : A_m \vdash Q :: (z : C_r)}
  \\[1em]
  \infer[{\up}R]
  {\Gamma \vdash \mb{case}\, x_m\, (\mb{shift}(y_k) \Rightarrow P) :: (x : \up_k^m A_k)}
  {\Gamma \vdash P :: (y : A_k)}
  \qquad
  \infer[{\up}L^0]
  {\Gamma_W, x : \up_k^m A_k \vdash x_m.\mb{shift}(y_k) :: (y : A_k)}
  {\mathstrut}
\end{crules}%
\caption{Typing rules ($\alpha \in \{0,1\}$ with $\alpha = 1$
  permitted only if $C \in \sigma(m)$)}
\label{fig:typing-rules}
\end{figure}

Now, our types are the propositions of adjoint logic, augmented with general equirecursive
types formed via mutually recursive type definitions in a global signature --- most of the
basic types are familiar as propositions from intuitionistic linear logic~\cite{Girard87tapsoft},
or as session types~\cite{Honda93concur,Honda98esop,Gay10jfp}, tagged with subscripts for modes.
The only new types are the shifts $\down^r_m A_r$ and $\up_k^m A_k$. We do,
however, change $\oplus$ and $\with$ slightly, using an $n$-ary form rather than
the standard binary form. These are, of course, equivalent, but the $n$-ary form
allows for more natural programming. The grammar of types (as well as processes)
can be found in \cref{fig:types-and-procs}. Note that while our grammar includes
mode subscripts on types, type constructors, and variables we will often omit
them when they are clear from context.

The typing judgment for processes has the form
\[
  x_1:A^1_{m_1}, \ldots, x_n:A^n_{m_n} \vdash P :: (x : A_k)
\]
where $P$ is a process expression and we require that each $m_i \geq k$. Given such a judgment, we say that $P$
\emph{provides} or \emph{writes} $x$, and \emph{uses} or \emph{reads} $x_1, \ldots, x_n$.
The rules for this judgment can be found in \cref{fig:typing-rules}, where we have
elided a fixed signature $\Sigma$ with type and process definitions explained
later in this section.  As usual, we require each of the $x_i$ and $x$ to be
distinct and allow silent renaming of bound variables in process expressions.

In this formulation, contraction and weakening remain \emph{implicit}.\footnote{See~\cite{Pruiksma18un} for a
  formulation of adjoint logic with \emph{explicit} structural rules, which are less amenable to programming.}
Handling contraction leads us to two versions of each of the $\oplus, \one, \tensor$ left rules,
depending on whether the principal formula of the rule can be used again or not. The subscript $\alpha$ on each
of these rules may be either $0$ or $1$, and indicates whether the principal formula of the rule is
preserved in the context. The $\alpha = 0$ version of each rule is the standard linear form,
while the $\alpha = 1$ version, which requires that the mode $m$ of the principal formula satisfies
$C \in \sigma(m)$, keeps a copy of the principal formula. Note that if $C \in \sigma(m)$, we are
still allowed to use the $\alpha = 0$ version of the rule. Moreover, we write $\Gamma_C, \Gamma_W$
for contexts of variables all of which allow contraction or weakening, respectively. This allows us
to freely drop weakenable variables when we reach initial rules, or to duplicate
contractable variables to both parent and child when spawning a new process in the $\cut$ rule.

Note that there is no explicit rule for (possibly recursively defined) type variables $t$, since they
can be silently replaced by their definitions. Equality between types and type-checking can both easily
be seen to be decidable using a combination of standard techniques for substructural type
systems~\cite{Cervesato00tcs} and subtyping for equirecursive session types,~\cite{Gay05acta} which
relies on a coinductive interpretation of the types, but not on their linearity, and so can be adapted
to the adjoint setting. Some experience with a closely related algorithm~\cite{Das20fscd}
for type equality and type checking suggests that this is practical.

We now go on to briefly examine the terms and loosely describe their meanings from the perspective
of a shared-memory semantics. We will make this more precise in
\cref{sec:sem-concur,sec:sem-seq}, where we develop the dynamics of such a shared-memory semantics.

Both the grammar and the typing rules show that we have five primary constructs for processes,
which then break down further into specific cases.

The first two process constructs are 
type-agnostic. The $\cut$ rule, with term $x \leftarrow P \semi Q$,
allocates a new memory cell $x$, spawns a new process $P$ which may write to $x$,
and continues as $Q$ which may read from $x$. The new cell $x$ thus serves as
the point of communication between the new process $P$ and the continuing $Q$.
The $\id$ rule, with term $x_m \leftarrow y_m$, copies the contents of cell $y_m$ into
cell $x_m$. If $C \notin \sigma(m)$, we can think of this instead as
\emph{moving} the contents of cell $y_m$ into cell $x_m$ and freeing $y_m$.

The next two constructs, $x.V$ and $\mb{case}\, x\, K$, come in pairs that perform communication,
one pair for each type. A process of one of these forms will either write to or read from
$x$, depending on whether the variable is in the succedent (write) or antecedent (read).

A write is straightforward and stores either the value
$V$ or the continuation $K$ into the cell $x$, while a read pulls a continuation $K'$ or a
value $V'$ from the cell, and combines either $K'$ and $V$ (in the case of $x.V$) or
$K$ and $V'$ ($\mb{case}\, x\, K$). The symmetry of this, in which continuations and values
are both eligible to be written to memory and read from memory, comes from the duality between
$\oplus$ and $\with$, between $\otimes$ and $\lolli$, and
between $\down$ and $\up$. We see this in the typing rules, where, for instance,
$\oplus R^0$ and $\with L^0$ have the same process term, swapping only the roles of
each variable between read and write.
As cells may contain
either values $V$ or continuations $K$, it will be useful to have a way to refer to
this class of expression:
\[ \mbox{Cell contents} \qquad W ::= V \mid K \]

The final construct allows for calling named processes, which we use for recursion.
As is customary in session types, we use \emph{equirecursive types}, collected in a
signature $\Sigma$ in which we also collect recursive process definitions and their
types. For each type definition $t = A$, the type $A$ must be
\emph{contractive} so that we can treat types \emph{equirecursively} with a
straightforward coinductive definition and an efficient algorithm for
type equality~\cite{Gay05acta}.

A named process $p$ is \emph{declared} as
$B^1_{m_1}, \ldots, B^n_{m_n} \vdash p :: A_k$ which means it requires
arguments of types $B^i_{m_i}$ (in that order) and provides a result
of type $A_k$.  For ease of readability, we may sometimes write in
variables names as well, but they are actually only needed for the
corresponding \emph{definitions}
$x \leftarrow p \larrow y_1, \ldots, y_n = P$.
Operationally, a call $z \leftarrow p \larrow \overline{w}$
expands to its definition with a substitution
$[\overline{w}/\overline{y}, z/x]P$, replacing variables by addresses.

We can then formally define signatures as follows, allowing definitions
of types, declarations of processes, and definitions of processes:
\[
  \begin{array}{llcl}
    \mbox{Signatures} & \Sigma & ::= & \cdot \mid \Sigma, t = A
                                      \mid \Sigma, \overline{B_m} \vdash p :: A_k
                                      \mid \Sigma, x \leftarrow p \larrow \overline{y} = P
  \end{array}
\]
For valid signatures we require that each declaration
$\overline{B_m} \vdash p :: A_k$ has a corresponding definition
$x \leftarrow p \larrow \overline{y} = P$ with
$\overline{y : B_m} \vdash P :: (x : A_k)$.  This means that all
type and process definitions can be mutually recursive.

In the remainder of this paper we assume that we have a fixed valid
signature $\Sigma$, so we annotate neither the typing judgment nor the
computation rules with an explicit signature.

\section{Concurrent Semantics}
\label{sec:sem-concur}

We will now present a concurrent shared-memory semantics for {\langname},
using \emph{multiset rewriting rules}~\cite{Cervesato09ic}.
The state of a running program is a multiset of semantic objects, which
we refer to as a \emph{process configuration}. We have three distinct types of
semantic objects:
\begin{enumerate}
\item $\m{thread}(c_m, P)$: thread executing $P$ with destination $c_m$
\item $\m{cell}(c_m, \uscore)$: cell $c_m$ that has been allocated, but not yet written
\item $\bang_m \m{cell}(c_m, W)$: cell $c_m$ containing $W$
\end{enumerate}
Here, we prefix a semantic object with $\bang_m$ to indicate that it is persistent
when $C \in \sigma(m)$, and ephemeral otherwise. Note that empty cells are always
ephemeral, so that we can modify them by writing to them, while filled cells may be
persistent, as each cell has exactly one writer, which will terminate on writing.
We maintain the invariant that in a configuration either
$\m{thread}(c_m, P)$ appears together with $\bang_m \m{cell}(c_m, \uscore)$, or we
have just $\bang_m \m{cell}(c_m, W)$, as well as that if two semantic objects provide
the same address $c_m$, then they are exactly a
$\m{thread}(c_m, P)$, $\bang_m \m{cell}(c_m, \uscore)$ pair.
While this invariant can be made slightly cleaner by removing the $\bang_m \m{cell}(c_m, \uscore)$
objects, this leads to an interpretation where cells are allocated lazily just before
they are written. While this has some advantages, it is unclear how to inform the
thread which will eventually \emph{read from} the new cell where said cell can be found,
and so, in the interest of having a realistically implementable semantics, we just allocate
an empty cell on spawning a new thread, allowing the parent thread to see the location of that
cell.

We can then define configurations with the following grammar (and the additional constraint
of our invariant):
\[
  \begin{array}{llcl}
    \mbox{Configurations} & \CC & ::= & \cdot \mid \m{thread}(c_m, P), \m{cell}(c_m, \uscore)
                                       \mid \bang_m \m{cell}(c_m, W) \mid \CC_1, \CC_2
  \end{array}
\]
We think of the join $\CC_1, \CC_2$ of two configurations as a commutative and associative
operation so that this grammar defines a multiset rather than a list or tree.

A multiset rewriting rule takes the collection of objects on the
left-hand side of the rule, consumes them (if they are ephemeral), and then
adds in the objects on the right-hand side of the rule. Rules may
be applied to any subconfiguration, leaving the remainder of the configuration
unchanged. This yields a naturally nondeterministic semantics, but we will see
that the semantics are nonetheless confluent (\Cref{thm:diamond}).
Additionally, while our configurations are not ordered, we will adopt the
convention that the writer of an address appears to the left of any readers
of that address.

Our semantic rules are based on a few key ideas:
\begin{enumerate}
\item Variables represent addresses in shared memory.
\item Cut/spawn is the only way to allocate a new cell.
\item Identity/forward will move or copy data between cells.
\item A process $\m{thread}(c, P)$ will (eventually) write to
  the cell at address $c$ and then terminate.
\item A process $\m{thread}(d, Q)$ that is trying to read from
  $c \neq d$ will wait until the cell with address $c$
  is available (i.e. its contents is no longer $\uscore$),
  perform the read, and then continue.
\end{enumerate}

The counterintuitive part of this interpretation (when using a message-passing semantics as a point
of reference) is that a process providing $c : A \with B$ \emph{does not read a value from shared
  memory}. Instead, it writes a continuation to memory and terminates. Conversely, a client of such
a channel \emph{does not write a value to shared memory}. Instead, it continues by jumping to the
continuation. This ability to write continuations to memory is a major feature distinguishing this
from a message-passing semantics where potentially large closures would have to be captured,
serialized, and deserialized, the cost of which is difficult to control~\cite{Miller16onward}.

The final piece that we need to present the semantics is a key operation, namely that of passing a
value $V$ to a continuation $K$ to get a new process $P$. This operation is defined as follows:
\[
  \begin{array}{lcl@{\quad}l}
    i(d) \circ (\ell(y) \Rightarrow P_\ell)_{\ell \in L} & = & [d/y]P_i & (\oplus,\with) \\
    \langle e,c\rangle \circ (\langle w,y\rangle \Rightarrow P) & = & [e/w,c/y]P & (\tensor,\lolli) \\
    \langle\, \rangle \circ (\langle\, \rangle \Rightarrow P) & = & P & (\one) \\
    \mb{shift}(d) \circ  (\mb{shift}(y) \Rightarrow P)  & = & [d/y]P & (\down, \up) \\
  \end{array}
\]
When any of these reductions is applied, either the value or the continuation
has been read from a cell while the other is a part of the executing process.
With this notation, we can give a concise set of rules for the concurrent
dynamics. 
We present these rules in \cref{fig:dyn-concur}.

\begin{figure}[htb]
\[\begin{small}
  \begin{array}{ll}
    \m{thread}(c, x \leftarrow P \semi Q) \mapsto \m{thread}(a, [a/x]P), \m{cell}(a, \uscore),
    \m{thread}(c, [a/x]Q)
    \; \mbox{($a$ fresh)} & \mbox{\it cut: allocate \& spawn} \\[1ex]
    \bang_m \m{cell}(c_m,W), \m{thread}(d_m, d_m \leftarrow c_m), \m{cell}(d_m, \uscore)
    \mapsto \bang_m \m{cell}(d_m,W) & \mbox{\it id: move or copy} \\[1ex]
    \m{thread}(c, c \leftarrow p \larrow \overline{d}) \mapsto
    \m{thread}(c, [c/x, \overline{d/y}]P) \quad
    \mbox{for $x \leftarrow p \larrow \overline{y} = P \in \Sigma$}
    & \mbox{\it call} \\[1ex]
    \m{thread}(c_m, c_m.V), \m{cell}(c_m, \uscore) \mapsto \bang_m \m{cell}(c_m, V)
    & ({\oplus}R^0,{\tensor}R^0,{\one}R^0, {\down}R^0) \\
    \bang_m \m{cell}(c_m, V), \m{thread}(e_k, \mb{case}\, c_m\, K)
    \mapsto \m{thread}(e_k, V\circ K) & ({\oplus}L,{\tensor}L,{\one}L,{\down}L) \\[1ex]
    \m{thread}(c_m, \mb{case}\, c_m\, K), \m{cell}(c_m, \uscore)
    \mapsto \bang_m \m{cell}(c_m, K) & ({\lolli}R, {\with}R, {\up}R) \\
    \bang_m \m{cell}(c_m, K), \m{thread}(d_k, c_m.V) \mapsto \m{thread}(d_k, V\circ K)
    &({\lolli}L^0,{\with}L^0,{\up}L^0)
  \end{array}
  \end{small}
\]
\caption{Concurrent dynamic rules}
\label{fig:dyn-concur}
\end{figure}

These rules match well with our intuitions from before. In the cut rule, we allocate a
new empty cell $a$, spawn a new thread to execute $P$, and continue executing $Q$, just as we
described informally in \cref{sec:type-system}. Similarly, in the id rule, we either
move or copy (depending on whether $C \in \sigma(m)$) the 
contents of cell $c$ into cell $d$ and terminate.
The rules that write values to cells are exactly the right rules for positive types ($\oplus, \otimes, \one, \down$),
while the right rules for negative types ($\with, \lolli, \up$) write continuations to cells
instead. Dually, to read from a cell of positive type, we must have a continuation to pass the
stored value to, while to read from a cell of negative type, we need a value to pass to the
stored continuation.

\subsection{Results}


We have standard results for this system --- a form of progress, of preservation, and
a confluence result. To discuss progress and preservation, we must first extend our
notion of typing for process terms to configurations.
Configurations are typed with the judgment $\Gamma \vdash \CC :: \Delta$ which
means that configuration $\CC$ may read from the addresses in $\Gamma$ and
write to the addresses in $\Delta$.
We can then give the following set of rules for typing configurations, which make use
of the typing judgment $\Gamma \vdash P :: (c : A_m)$ for process terms in the base cases.
\begin{rules}
\infer[]
  {\Gamma_C, \Gamma, \Delta \vdash \m{thread}(c,P), \m{cell}(c, \uscore) :: (\Gamma_C, \Delta, c : A_m)}
  {\Gamma_C, \Gamma \vdash P :: (c : A_m)}
\\[1em]
\infer[]
  {\Gamma_C, \Gamma, \Delta \vdash \bang_m \m{cell}(c,V) :: (\Gamma_C, \Delta, c : A_m)}
  {\Gamma_C, \Gamma \vdash c.V :: (c : A_m)}
  \qquad
\infer[]
  {\Gamma_C, \Gamma, \Delta \vdash \bang_m \m{cell}(c,K) :: (\Gamma_C, \Delta, c : A_m)}
  {\Gamma_C, \Gamma \vdash \mb{case}\, c\, K :: (c : A_m)}
\\[1em]
\infer[]
  {\Delta \vdash (\cdot) :: \Delta}
  {\mathstrut}
\qquad
\infer[]
  {\Gamma \vdash \CC_1, \CC_2 :: \Delta_2}
  {\Gamma \vdash \CC_1 :: \Delta_1
  & \Delta_1 \vdash \CC_2 :: \Delta_2}
\end{rules}%
Note that our invariants on configurations mean that there is no need to separately
type the objects $\m{thread}(c, P)$ and $\m{cell}(c, \uscore)$, as they can only
occur together. Additionally, while our configurations are multisets, and therefore
not inherently ordered, observe that the typing derivation for a configuration induces
an order on the configuration, something which is quite useful in proving progress.
\footnote{This technique is used in \cite{DeYoung20fscd} to prove progress for
a similar language.}

Our preservation theorem differs slightly from the standard, in that
it allows the collection of typed channels $\Delta$ offered by a configuration
$\CC$ to grow after a step, as steps may introduce new persistent memory cells.
\begin{theorem}[Type Preservation]
  If $\Gamma \vdash \CC :: \Delta$ and $\CC \mapsto \CC'$
  then $\Gamma \vdash \CC' :: \Delta'$ for some $\Delta' \supseteq \Delta$.
\end{theorem}
\begin{proof}
  By cases on the transition relation for configurations, applying
  repeated inversions to the typing judgment on $\CC$ to obtain the
  necessary information to assemble a typing derivation for $\CC'$.
  This requires some straightforward lemmas expressing that
  non-interfering processes and cells can be exchanged in a typing
  derivation.
\end{proof}

Progress is entirely standard, with configurations comprised entirely of
filled cells taking the role that values play in a functional language.
\begin{theorem}[Progress]
  If $\cdot \vdash \CC :: \Delta$ then either
  \begin{enumerate}
  \item[(i)] $\CC \mapsto \CC'$ for some $\CC'$, or
  \item[(ii)] for every channel $c_m : A_m \in \Delta$ there is an
    object $\bang_m \m{cell}(c_m,W) \in \CC$.
  \end{enumerate}
\end{theorem}
\begin{proof}
  We first re-associate all applications of the typing rule for
  joining configurations to the left.  Then we perform an induction
  over the structure of the resulting derivation, distinguishing cases
  for the rightmost cell or thread and potentially applying the
  induction hypothesis on the configuration to its left.  This
  structure, together with inversion on the typing of the cell or
  thread yields the theorem.
\end{proof}

In addition to these essential properties, we also have a confluence
result, for which we need to define a weak notion of equivalence on
configurations. We say $\CC_1 \sim \CC_2$
if there is a renaming $\rho$ of addresses such that
$\rho\CC_1 = \CC_2$. We can then establish the following version
of the diamond property:
\begin{theorem}[Diamond Property]
\label{thm:diamond}
  Assume $\Delta \vdash \CC :: \Gamma$.  If
  $\CC \mapsto \CC_1$ and $\CC \mapsto \CC_2$ such
  that $\CC_1 \not\sim \CC_2$.  Then there exist $\CC_1'$ and $\CC_2'$ such that
  $\CC_1 \mapsto \CC_1'$ and $\CC_2 \mapsto \CC_2'$
  with $\CC_1' \sim \CC_2'$.
\end{theorem}
\begin{proof}
  The proof is straightforward by cases.  There are no critical pairs
  involving ephemeral (that is, non-persistent) objects in the
  left-hand sides of transition rules.
\end{proof}


\subsection{Examples}
\label{sec:examples}


We present here a few examples of concurrent programs, illustrating
various aspects of our language.

\paragraph{Example: Binary Numbers.}
As a first simple example we consider binary numbers, defined as a
type $\mi{bin}$ at mode $m$.  The structural properties of mode $m$
are arbitrary for our examples.  For concreteness, assume that $m$ is
\emph{linear}, that is, $\sigma(m) = \{\,\}$.
\[
  \mi{bin}_m = {\oplus}_m\{\m{b0} : \mi{bin}_m, \m{b1} : \mi{bin}_m, \m{e} : \one_m\}
\]
Unless multiple modes are involved, we will henceforth omit the mode
$m$.  As an example, the number $6 = (110)_2$ would be represented by
a sequence of labels $\m{e}, \m{b1}, \m{b1}, \m{b0}$, chained together
in a linked list. The first cell in the list would contain the bit
$\m{b0}$. It has some address $c_0$, and also contains an address
$c_1$ pointing to the next cell in the list. Writing out the whole
sequence as a configuration we have
\begin{rules}
  \m{cell}(c_4, \langle\,\rangle),
  \m{cell}(c_3, \m{e}(c_4)),
  \m{cell}(c_2, \m{b1}(c_3)),
  \m{cell}(c_1, \m{b1}(c_2)),
  \m{cell}(c_0, \m{b0}(c_1))
\end{rules}%

\paragraph{Example: Computing with Binary Numbers.}
We implement a recursive process $\mi{succ}$ that reads the bits of a
binary number $n$ starting at address $y$ and writes the bits for the
binary number $n+1$ starting at $x$.  This process may block until the
input cell (referenced as $y$) has been written to; the output cells
are allocated one by one as needed.  Since we assumed the mode $m$ is
linear, the cells read by the $\mi{succ}$ process from will be
deallocated.

\begin{tabbing}
  $\mi{bin}_m = {\oplus}_m\{\m{b0} : \mi{bin}_m, \m{b1} : \mi{bin}_m, \m{e} : \one_m\}$ \\
  $(y : \mi{bin}) \vdash \mi{succ} :: (x : \mi{bin})$ \\
  $x \leftarrow \mi{succ} \larrow y =$ \\
  \qquad $\mb{case}\; y$ \= $(\, \m{b0}(y') \Rightarrow $ \= $x' \leftarrow (x' \leftarrow y') \semi$ \hspace{3em}
  \= \% \it alloc $x'$ and copy $y'$ to $x'$ \\
  \>\> $x.\m{b1}(x')$ \> \% \it write $\m{b1}(x')$ to $x$ \\
  \> $\mid \m{b1}(y') \Rightarrow $ \= $x' \leftarrow (x' \leftarrow \mi{succ} \larrow y') \semi$
  \> \% \it alloc $x'$ and spawn $\mi{succ}\; y'$ \\
  \>\> $x.\m{b0}(x')$ \> \% \it write $\m{b0}(x')$ to $x$ \\
  \> $\mid \m{e}(y') \Rightarrow $ \= $x'' \leftarrow (x'' \leftarrow y') \semi $ \> \% \it alloc $x''$ and copy $y'$ to $x''$ \\
  \>\> $x' \leftarrow x'.\m{e}(x'') \semi $ \> \% \it alloc $x'$ and write $\m{e}(x'')$ to $x'$ \\
  \>\> $x.\m{b1}(x')\,)$ \> \% \it write $\m{b1}(x')$ to $x$
\end{tabbing}
%
In this example and others we find certain repeating patterns.
Abbreviating these makes the code easier to read and also more compact
to write.  As a first simplification, we can use the following
shortcuts:
\[
  \begin{array}{lcll}
    x \leftarrow y \semi Q & \triangleq & x\leftarrow (x \leftarrow y) \semi Q \\
    x \leftarrow f \larrow \overline{y} \semi Q & \triangleq & x \leftarrow (x \leftarrow f \larrow \overline{y}) \semi Q
  \end{array}
\]
With these, the code for successor becomes
\begin{tabbing}
  $x \leftarrow \mi{succ} \larrow y =$ \\
  \qquad $\mb{case}\; y$ \= $(\, \m{b0}(y') \Rightarrow $ \= $x' \leftarrow y'\semi$ \hspace{3em}
  \= \% \it alloc $x'$ and copy $y'$ to $x'$ \\
  \>\> $x.\m{b1}(x')$ \> \% \it write $\m{b1}(x')$ to $x$ \\
  \> $\mid \m{b1}(y') \Rightarrow $ \= $x' \leftarrow \mi{succ} \larrow y' \semi$
  \> \% \it alloc $x'$ and spawn $\mi{succ}\; y'$ \\
  \>\> $x.\m{b0}(x')$ \> \% \it write $\m{b0}(x')$ to $x$ \\
  \> $\mid \m{e}(y') \Rightarrow $ \= $x'' \leftarrow y' \semi$ \> \% \it alloc $x''$ and copy $y'$ to $x''$ \\
  \>\> $x' \leftarrow x'.\m{e}(x'') \semi$ \> \% \it alloc $x'$ and write $\m{e}(x'')$ to $x'$ \\
  \>\> $x.\m{b1}(x')\,)$ \> \% \it write $\m{b1}(x')$ to $x$
\end{tabbing}
The second pattern we notice are sequences of allocations
followed by immediate (single) uses of the new address.
We can collapse these by a kind of specialized substitution.
We describe the inverse, namely how the abbreviated notation
is \emph{elaborated} into the language primitives.
\[
  \begin{array}{llcl}
    \mbox{Value Sequence} & \bar{V} & ::= & i(\bar{V}) \mid (y,\bar{V}) \mid \mb{shift}(\bar{V}) \mid V
  \end{array}
\]
At positive types (${\oplus},{\tensor},{\one},{\down}$), which
write to the variable $x$ with $x.\bar{V}$, we define:
\[
  \begin{array}{lcll}
    x\mathrel{.} i(\bar{V}) & \triangleq & x' \leftarrow x'\mathrel{.} \bar{V} \semi x.i(x') & (\oplus) \\
    x\mathrel{.} \langle y,\bar{V}\rangle & \triangleq & x' \leftarrow x'\mathrel{.} \bar{V} \semi x.\langle y,x'\rangle & (\tensor) \\
    x\mathrel{.} \mb{shift}(\bar{V}) & \triangleq & x' \leftarrow x'\mathrel{.} \bar{V} \semi x.\mb{shift}(x') & (\down) \\
  \end{array}
\]
In each case, and similar definitions below, $x'$ is a fresh variable.
Using these abbreviations in our example, we can shorten it further.
\begin{tabbing}
  $x \leftarrow \mi{succ} \larrow y =$ \\
  \qquad $\mb{case}\; y$ \= $(\, \m{b0}(y') \Rightarrow $ \= $x.\m{b1}(y')$ \hspace{3em}
  \= \% \it write $\m{b1}(y')$ to $x$ \\
  \> $\mid \m{b1}(y') \Rightarrow $ \= $x' \leftarrow \mi{succ} \larrow y' \semi$
  \> \% \it alloc $x'$ and spawn $\mi{succ}\; y'$ \\
  \>\> $x.\m{b0}(x')$ \> \% \it write $\m{b0}(x')$ to $x$ \\
  \> $\mid \m{e}(y') \Rightarrow $ \= $x.\m{b1}(\m{e}(y'))\,)$ \> \% \it write $\m{b1}(\m{e}(y'))$ to $x$
\end{tabbing}
For negative types (${\with},{\lolli},{\up}$) the expansion is
symmetric, swapping the left- and right-hand sides of the cut.  This
is because these constructs read a continuation from memory at $x$ and
pass it a value.
\[
  \begin{array}{lcll}
    x\mathrel{.} i(\bar{V}) & \triangleq & x' \leftarrow x.i(x') \semi x'\mathrel{.} \bar{V} & (\with) \\
    x\mathrel{.} \langle y,\bar{V}\rangle & \triangleq & x' \leftarrow x.\langle y,x'\rangle \semi x'\mathrel{.} \bar{V} & (\lolli) \\
    x\mathrel{.} \mb{shift}(\bar{V}) & \triangleq & x' \leftarrow x.\mb{shift}(x') \semi x'\mathrel{.} \bar{V} & (\up) 
  \end{array}
\]
Similarly, we can decompose a continuation matching against a value
sequence $(\bar{V} \Rightarrow P)$.  For simplicity, we assume here
that the labels for each branch of a pattern match for internal
($\oplus$) or external ($\with$) choice are distinct; a generalization
to nested patterns is conceptually straightforward but syntactically
somewhat complex so we do not specify it formally.
\[
  \begin{array}{lcll}
    (\ell(\bar{V}_\ell) \Rightarrow P_\ell)_{\ell \in L}
    & \triangleq & (\ell(x') \Rightarrow \mb{case}\, x'\, (\bar{V}_\ell \Rightarrow P_\ell))_{\ell \in L} & (\oplus,\with) \\
    (\langle y,\bar{V}\rangle \Rightarrow P)
    & \triangleq & (\langle y,x'\rangle \Rightarrow \mb{case}\, x'\, (\bar{V} \Rightarrow P)) & (\tensor,\lolli) \\
    (\mb{shift}(\bar{V}) \Rightarrow P)
    & \triangleq & (\mb{shift}(x') \Rightarrow \mb{case}\, x'\, (\bar{V} \Rightarrow P)) & (\down,\up) 
  \end{array}
\]
For example, we can rewrite the successor program one more time
to express that $y'$ in the last case must actually contain
the unit element $\langle\,\rangle$ and match against it as
well as construct it on the right-hand side.
\begin{tabbing}
  $x \leftarrow \mi{succ} \larrow y =$ \\
  \qquad $\mb{case}\; y$ \= $(\, \m{b0}(y') \Rightarrow $ \= $x.\m{b1}(y')$ \hspace{3em}
  \= \% \it write $\m{b1}(y')$ to $x$ \\
  \> $\mid \m{b1}(y') \Rightarrow $ \= $x' \leftarrow \mi{succ} \larrow y' \semi$
  \> \% \it alloc $x'$ and spawn $\mi{succ}\; y'$ \\
  \>\> $x.\m{b0}(x')$ \> \% \it write $\m{b0}(x')$ to $x$ \\
  \> $\mid \m{e}\, \langle\,\rangle \Rightarrow $ \= $x.\m{b1}(\m{e}\, \langle\,\rangle)\,)$ \> \% \it write $\m{b1}(\m{e}\, \langle\,\rangle)$ to $x$
\end{tabbing}
We have to remember, however, that intermediate matches and
allocations still take place and the last two programs are \emph{not
  equivalent} in case the process with destination $y'$ does not
terminate.

To implement $\mi{plus2}$ we can just compose $\mi{succ}$ with itself.
\begin{tabbing}
  $(z : \mi{bin}) \vdash \mi{plus2} :: (x : \mi{bin})$ \\[1ex]
  $x \leftarrow \mi{plus2} \larrow z =$ \\
  \qquad $y \leftarrow \mi{succ}\; z \semi$ \\
  \qquad $x \leftarrow \mi{succ}\; y$
\end{tabbing}
In our concurrent semantics, the two successor processes form a
concurrently executing pipeline --- the first reads the initial number
from memory, bit by bit, and then writes a new number (again, bit by
bit) to memory for the second successor process to read.

\paragraph{Example: MapReduce.}
As a second example we consider $\mi{mapReduce}$ applied to a tree.
We have a neutral element $z$ (which stands in for every leaf) and a
process $f$ to be applied at every node to reduce the whole tree to a
single value.  This exhibits a high degree of parallelism, since the
operations on the left and right subtree can be done independently.
We abstract over the type of element $A$ and the result $B$ at the
meta-level, so that $\m{tree}_A$ is a family of types, and
$\mi{mapReduce}_{AB}$ is a family of processes, indexed by $A$ and
$B$.
\begin{tabbing}
  $\mi{tree}_A = {\oplus}_m\{\m{empty} : \one, \m{node} : \mi{tree}_A \tensor A \tensor \mi{tree}_A\}$
\end{tabbing}
Since $\mi{mapReduce}$ applies reduction at every node in the
tree, it is \emph{linear} in the tree.  On the other hand,
the neutral element $z$ is used for every leaf, and the
associative operation $f$ for every node, so $z$ requires
at least contraction (there must be at least one leaf)
and $f$ both weakening and contraction (there may be arbitrarily
many nodes).  Therefore we use three modes: the linear mode $m$
for the tree and the result of $\mi{mapReduce}$, a strict
mode $s$ for the neutral element $z$, and an unrestricted
mode $u$ for the operation applied at each node.

\begin{tabbing}
  $(z : \up^s_m B)\, (f : \up^u_m (B \tensor A \tensor B \lolli B))\, (t : \mi{tree}_A)
  \vdash \mi{mapReduce}_{AB} :: (s : B)$ \\[1ex]
  $s \leftarrow \mi{mapReduce}_{AB} \larrow z\; f\; t =$ \\
  \qquad $\mb{case}\; t\; $ \= $(\, \m{empty}\,\langle\,\rangle \Rightarrow $ \= $ z' \leftarrow z.\mb{shift}(z')$
  \hspace{7em} \= \% \it drop $f$ \\
  \>\> $s \leftarrow z'$ \\
  \> $\mid \m{node}\,\langle l,\langle x,r\rangle\rangle \Rightarrow $ \= $l' \leftarrow \mi{mapReduce}_{AB}\; z\; f\; l \semi$ \\
  \>\> $r' \leftarrow \mi{mapReduce}_{AB}\; z\; f\; r \semi$ \> \% \it duplicate $z$ and $f$ \\
  \>\> $p \leftarrow p.\langle l',\langle x,r'\rangle\rangle \semi$ \\
  \>\> $s' \leftarrow f.\mb{shift}\langle p,s'\rangle \semi$ \\
  \>\> $s \leftarrow s'\, )$
\end{tabbing}


\paragraph{Example: $\lambda$-Calculus.}
As a third example we show an encoding of the $\lambda$-calculus
using higher-order abstract syntax and parallel evaluation.
We specify, at an arbitrary mode $m$:
\begin{tabbing}
  $\mi{exp}_m = \oplus\{\m{app} : \mi{exp} \tensor \mi{exp}, \m{val} : \mi{val}\}$ \\
  $\mi{val}_m = \oplus\{\m{lam} : \mi{val} \lolli \mi{exp}\}$
\end{tabbing}
An interesting property of this representation is that if we pick $m$
to be linear, we obtain the linear $\lambda$-calculus~\cite{Lincoln92lics},
if we pick $m$ to be strict ($\sigma(m) = \{C\}$) we
obtain Church and Rosser's original $\lambda I$
calculus~\cite{Church36}, and if we set $\sigma(m) = \{W,C\}$ we obtain
the ``usual'' $\lambda$-calculus~\cite{Barendregt84book}.  Evaluation (that
is, parallel reduction to a weak head-normal form) is specified by the
following process, no matter which version of the $\lambda$-calculus
we consider.
\begin{tabbing}
  $(e : \mi{exp}) \vdash \mi{eval} : (v : \mi{val})$ \\[1ex]
  $v \leftarrow \mi{eval}\larrow e =$ \\
  \qquad $\mb{case}\; e$ \= $(\, \m{val}(v') \Rightarrow v \leftarrow v'$ \\
  \> $\mid \m{app}\,\langle e_1,e_2\rangle \Rightarrow $ \= $v_1 \leftarrow \mi{eval} \larrow e_1 \semi$ \\
  \>\> $v_2 \leftarrow \mi{eval} \larrow e_2 \semi$ \\
  \>\> $\mb{case}\; v_1$ $(\, \m{lam}(f) \Rightarrow $ \= $e_1' \leftarrow f.\langle v_2,e_1'\rangle \semi$
  \hspace{1em} \% \it $f : \mi{val} \lolli \mi{exp}$ \\
  \>\>\> $v \leftarrow \mi{eval} \larrow e_1'\, )\, )$
\end{tabbing}
In this code, $v_2$ acts like a \emph{future}: we spawn the evaluation
of $e_2$ with the promise to place the result in $v_2$.  In our
dynamics, we allocate a new cell for $v_2$, as yet unfilled.  When we
pass $v_2$ to $f$ in $f.\langle v_2,e_1'\rangle$ the process
$\mi{eval}\; e_2$ may still be computing and we will not block until
we eventually try to read from $v_2$ (which may nor may not happen).

\section{Sequential Semantics}
\label{sec:sem-seq}

While our concurrent semantics is quite expressive and allows for a great
deal of parallelism, in a real-world setting, the overhead of spawning a
new thread can make it inefficient to do so unless the work that thread
does is substantial. Moreover, many of the patterns of concurrent
computation that we would like to model involve adding some limited access
to concurrency in a largely sequential language. We can address both of
these issues with the concurrent semantics by adding a construct to enforce
sequentiality. Here, we will take as our definition of sequentiality that
only one thread (the \emph{active thread}) is able to take a step at a time,
with all other threads being blocked.

The key idea in enforcing sequentiality is to observe that only the cut/spawn
rule turns a single thread into two. When we apply the cut/spawn rule to the
term $x \leftarrow P \semi Q$, $P$ and $Q$ are executed concurrently. One obvious
way (we discuss another later in this section) to enforce
sequentiality is to introduce a \emph{sequential cut} construct
$x \Leftarrow P \semi Q$ that ensures that $P$ runs to completion,
writing its result into $x$, before $Q$ can continue. We do not believe that
we can ensure this using our existing (concurrent) semantics. However, with 
a small addition to the language and semantics, we are able to define a
sequential cut as syntactic sugar for a {\langname} term that does enforce this.

\paragraph{Example Revisited: A Sequential Successor.}
Before we move to the formal definition that enforces sequentiality, we
reconsider the successor example on binary numbers in its most explicit form.
We make all cuts sequential.
\begin{tabbing}
  $\mi{bin}_m = {\oplus}_m\{\m{b0} : \mi{bin}_m, \m{b1} : \mi{bin}_m, \m{e} : \one_m\}$ \\
  $(y : \mi{bin}) \vdash \mi{succ} :: (x : \mi{bin})$ \\
  $x \leftarrow \mi{succ} \larrow y =$ \\
  \qquad $\mb{case}\; y$ \= $(\, \m{b0}(y') \Rightarrow $ \= $x' \Leftarrow (x' \leftarrow y') \semi$ \hspace{3em}
  \= \% \it alloc $x'$ and copy $y'$ to $x'$ \\
  \>\> $x.\m{b1}(x')$ \> \% \it write $\m{b1}(x')$ to $x$ \\
  \> $\mid \m{b1}(y') \Rightarrow $ \= $x' \Leftarrow (x' \leftarrow \mi{succ} \larrow y') \semi$
  \> \% \it alloc $x'$ and spawn $\mi{succ}\; y'$ \\
  \>\> $x.\m{b0}(x')$ \> \% \it write $\m{b0}(x')$ to $x$ \\
  \> $\mid \m{e}(y') \Rightarrow $ \= $x'' \Leftarrow (x'' \leftarrow y')$ \> \% \it alloc $x''$ and copy $y'$ to $x''$ \\
  \>\> $x' \Leftarrow x'.\m{e}(x'')$ \> \% \it alloc $x'$ and write $\m{e}(x'')$ to $x'$ \\
  \>\> $x.\m{b1}(x')\,)$ \> \% \it write $\m{b1}(x')$ to $x$
\end{tabbing}
This now behaves like a typical sequential implementation of a
successor function, but in destination-passing
style~\cite{Wadler84slfp,Larus89thesis,Cervesato02tr,Simmons12phd}.
When there is a carry (manifest as a recursive call to $\mi{succ}$),
the output bit zero will not be written until the effect of the carry
has been fully computed.

To implement sequential cut, we will take advantage of the fact that a shift from
a mode $m$ to itself does not affect provability, but does force synchronization.
If $x : A_m$, we would like to define
\[
x \Leftarrow P \semi Q \triangleq x' \leftarrow P' \semi \mb{case}\; x'\; (\mb{shift}(x) \Rightarrow Q),
\]
where $x' : \down^m_m A_m$, and (informally) $P'$ behaves like $P$, except that
wherever $P$ would write to $x$, $P'$ writes simultaneously to $x$ and $x'$.
Setting aside for now a formal definition of $P'$, we see that $Q$ is blocked until
$x'$ has been written to, and so as long as $P'$ writes to $x'$ no later than it
writes to $x$, this ensures that $x$ is written to before $Q$ can continue.

We now see that in order to define $P'$ from $P$, we need some way to ensure that
$x'$ is written to no later than $x$. The simplest way to do this is to add a form
of \emph{atomic write} which writes to two cells simultaneously. We define three new
constructs for these atomic writes, shown here along with the non-atomic processes
that they imitate. We do not show typing rules here, but each atomic write can be typed
in the same way as its non-atomic equivalent.
\[
\begin{tabular}{l@{\quad}r}
\hline
Atomic Write & Non-atomic equivalent \\
\hline
$x'.\mb{shift}(x.V)$ & 
$x \leftarrow x.V \semi x'.\mb{shift}(x)$ \\
$x'.\mb{shift}(\mb{case}\; x\; K)$ &
$x \leftarrow \mb{case}\; x\; K \semi x'.\mb{shift}(x)$ \\
$x'.\mb{shift}(x \leftarrow y)$ &
$x \leftarrow (x \leftarrow y) \semi x'.\mb{shift}(x)$ \\
\hline
\end{tabular}
\]
Each atomic write simply evaluates in a single step to the configuration where both
$x$ and $x'$ have been written to, much as if the non-atomic equivalent had taken
three steps --- first for the cut, second to write to $x$, and third to write to $x'$.
This intuition is formalized in the following transition rules:
\[\begin{small}
  \begin{array}{ll}
    \m{thread}(x'_m, x'_m.\mb{shift}(x_k.V)) \mapsto \bang_k \m{cell}(x_k, V), \bang_m \m{cell}(x'_m, \mb{shift}(x_k))
    & \mbox{\it atom-val} \\
    \m{thread}(x'_m, x'_m.\mb{shift}(\mb{case}\; x_k\; K)) \mapsto \bang_k \m{cell}(x_k, K), \bang_m \m{cell}(x'_m, \mb{shift}(x_k))
    & \mbox{\it atom-cont} \\
    \m{thread}(x'_m, x'_m.\mb{shift}(x_k \leftarrow y_k)), \bang_k \m{cell}(y_k, W) \mapsto
    \bang_k \m{cell}(x_k, W), \bang_m \m{cell}(x'_m, \mb{shift}(x_k))
    & \mbox{\it atom-id} \\
  \end{array}
  \end{small}
\]
Note that the rule for the identity case is different from the other two ---
it requires the cell $y_k$ to have been written to in order to continue. This
is because the $x \leftarrow y$ construct reads from $y$ and writes to $x$ ---
if we wish to write to $x$ and $x'$ atomically, we must also perform the read
from $y$. 

Now, to obtain $P'$ from $P$, we define a substitution operation
$[x'.\mb{shift}(x) /\!\!/ x]$ that replaces writes to $x$ with atomic writes to $x$ and
$x'$ as follows:
\[
\begin{array}{rcl}
(x.V)[x'.\mb{shift}(x) /\!\!/ x] &=& x'.\mb{shift}(x.V) \\
(\mb{case}\; x\; K)[x'.\mb{shift}(x) /\!\!/ x] &=& x'.\mb{shift}(\mb{case}\; x\; K) \\
(x \leftarrow y)[x'.\mb{shift}(x) /\!\!/ x] &=& x'.\mb{shift}(x \leftarrow y) \\
\end{array}
\]
Extending $[x'.\mb{shift}(x) /\!\!/ x]$ compositionally over our other language constructs,
we can define $P' = P[x'.\mb{shift}(x) /\!\!/ x]$, and so
\[
x \Leftarrow P \semi Q \triangleq x' \leftarrow P[x'.\mb{shift}(x) /\!\!/ x] \semi \mb{case}\; x'\; (\mb{shift}(x) \Rightarrow Q).
\]

We now can use the sequential cut to enforce an order on computation. Of particular
interest is the case where we restrict our language so that all cuts are sequential.
This gives us a fully sequential language, where we indeed have that only one
thread is active at a time.
We will make extensive use of this ability to give a fully sequential language,
and in \cref{sec:fork-join,sec:futures,sec:sill}, we will add back limited
access to concurrency to such a sequential language in order to reconstruct
various patterns of concurrent computation.

There are a few properties of the operation $[x'.\mb{shift}(x) /\!\!/ x]$
and the sequential cut that we will make use of in our embeddings.
Essentially, we would like
to know that $P[x'.\mb{shift}(x) /\!\!/ x]$ has similar behavior
from a typing perspective to $P$, and that a sequential
cut can be typed in a similar manner to a standard concurrent cut.
We make this precise with the following lemmas:
\begin{lemma}
\label{lem:seq-shift-typing}
If $\Gamma \vdash P :: (x : A_m)$, then $\Gamma \vdash P[x'.\mb{shift}(x) /\!\!/ x] :: (x' : \down^m_m A_m)$.
\end{lemma}
\begin{lemma}
\label{lem:seqcut}
The rule
\begin{srules}
\infer[\m{seqcut}]
  {\Gamma_C, \Delta, \Delta' \vdash (x \Leftarrow P \semi Q) :: (z : C_r)}
  {(\Gamma_C, \Delta \geq m \geq r)
    & \Gamma_C, \Delta \vdash P :: (x : A_m)
    & \Gamma_C, \Delta', x : A_m \vdash Q :: (z : C_r)
  }
\end{srules}
is admissible.
\end{lemma}
\Cref{lem:seq-shift-typing} follows from a simple induction on the
structure of $P$, and \cref{lem:seqcut} can be proven by deriving
the $\m{seqcut}$ rule using \cref{lem:seq-shift-typing}.

In an earlier version of this
paper,\footnote{\cite{Pruiksma20arxiv}, version 1, found at https://arxiv.org/abs/2002.04607}
we developed a separate set of sequential semantics which is bisimilar to the
presentation we give here in terms of sequential cuts. However, by embedding
the sequential cut into the concurrent semantics as syntactic sugar, we are
able to drastically reduce the conceptual and technical overhead needed to
look at interactions between the two different frameworks, and simplify our
encodings of various concurrency patterns.

\paragraph{Example Revisited: $\lambda$-Calculus.}
If we make all cuts in the $\lambda$-calculus interpreter sequential, we obtain
a \emph{call-by-value} semantics.  In particular, it may no longer compute
a weak head-normal form even if it exists.
\begin{tabbing}
  $\mi{exp}_m = \oplus\{\m{app} : \mi{exp} \tensor \mi{exp}, \m{val} : \mi{val}\}$ \\
  $\mi{val}_m = \oplus\{\m{lam} : \mi{val} \lolli \mi{exp}\}$ \\[1ex]
  $(e : \m{exp}) \vdash \mi{eval} : (v : \m{val})$ \\[1ex]
  $v \leftarrow \mi{eval}\larrow e =$ \\
  \qquad $\mb{case}\; e$ \= $(\, \m{val}(v') \Rightarrow v \leftarrow v'$ \\
  \> $\mid \m{app}\,\langle e_1,e_2\rangle \Rightarrow $ \= $v_1 \Leftarrow \mi{eval} \larrow e_1 \semi$ \\
  \>\> $v_2 \Leftarrow \mi{eval} \larrow e_2 \semi$ \\
  \>\> $\mb{case}\; v_1$ $(\, \m{lam}(f) \Rightarrow $ \= $e_1' \Leftarrow f.\langle v_2,e_1'\rangle \semi$ \\
  \>\>\> $v \Leftarrow \mi{eval} \larrow e_1'\, )\, )$
\end{tabbing}

\paragraph{Call-by-name.}
\label{par:call-by-name}

As mentioned at the beginning of this section, there are multiple approaches to
enforcing that only one thread is active at a time. We can think of the sequential cut
defined in \cref{sec:sem-seq} as a form of \emph{call-by-value} --- $P$ is fully
evaluated before $Q$ can continue. Here, we will define a different sequential cut
$x \Leftarrow^N P \semi Q$, which will behave more like \emph{call-by-name}, delaying
execution of $P$ until $Q$ attempts to read from $x$.
Interestingly, this construct avoids the need for atomic write operations!
We nevertheless prefer the ``call-by-value'' form of sequentiality as our default,
as it aligns better with Halstead's approach to futures,
which were defined in a call-by-value language, and also avoids recomputing $P$ if $x$
is used multiple times in $Q$.

As before, we take advantage of shifts for synchronization, here using an upwards shift
rather than a downwards one. If $x : A_m$, we would like to define
\[
x \Leftarrow^N P \semi Q \triangleq x' \leftarrow \mb{case}\; x'\;(\mb{shift}(x) \Rightarrow P) \semi Q',
\]
where $x' : \up_m^m A_m$ and $Q'$ behaves as $Q$, except that where $Q$ would read from $x$, $Q'$
first reads from $x'$ and then from $x$. We can formalize the operation that takes $Q$ to $Q'$ in
a similar manner to $[x'.\mb{shift}(x) /\!\!/ x]$. We will call this operation
$[x'.\mb{shift}(x) \% x]$, so $Q' = Q[x'.\mb{shift}(x) \% x]$.
\[
\begin{array}{rcl}
(x.V)[x'.\mb{shift}(x) \% x] &=& x \leftarrow x'.\mb{shift}(x) \semi x.V \\
(\mb{case}\; x\; K)[x'.\mb{shift}(x) \% x] &=& x \leftarrow x'.\mb{shift}(x) \semi \mb{case}\; x\; K \\
(y \leftarrow x)[x'.\mb{shift}(x) \% x] &=& x \leftarrow x'.\mb{shift}(x) \semi y \leftarrow x \\
\end{array}
\]
Note that unlike in our ``call-by-value'' sequential cut, where we needed to write to
two cells atomically, here, the order of reads is enforced because $x'.\mb{shift}(x)$
will execute the stored continuation $\mb{shift}(x) \Rightarrow P$, which finishes by
writing to $x$. As such, we are guaranteed that $Q'$ is paused waiting to read from
$x$ until $P$ finishes executing. Moreover, $P$ is paused within a continuation until
$Q'$ reads from $x'$, after which it immediately blocks on $x$, so we maintain only one
active thread as desired.

While we will not make much use of this form of sequentiality, we find it interesting
that it is so simply encoded, and that the encoding is so similar to that of call-by-value
cuts. Both constructions are also quite natural --- the main decision that we make is
whether to pause $P$ or $Q$ inside a continuation. From this, the rest of the construction
follows, as there are two natural places to wake up the paused process --- as early as possible
or as late as possible. If we wake the paused process $P$ immediately after the cut, as in
\[
x' \leftarrow \mb{case}\; x'\;(\mb{shift}(x) \Rightarrow P) \semi x \leftarrow x'.\mb{shift}(x) \semi Q,
\]
the result is a concurrent cut with the extra overhead of the shift. Our sequential
cuts are the result of waking the paused process as late as possible --- once there is
no more work to be done in $P$ in the call-by-value cut, and once $Q$ starts to actually
depend on the result of $P$ in the call-by-name cut.

\paragraph{$\lambda$-Calculus Example Revisited.}
We can achieve a sequential interpreter for the $\lambda$-calculus with a single use of a
by-name cut.  This interpreter is then complete: if a weak head-normal form exists, it
will compute it.  We also recall that this property holds no matter which structural
properties we allow for the $\lambda$-calculus (e.g., purely linear if the mode
allows neither weakening nor contraction, of the $\lambda I$-calculus if the mode
only allows contraction).
\begin{tabbing}
  $v \leftarrow \mi{eval}\larrow e =$ \\
  \qquad $\mb{case}\; e$ \= $(\, \mi{val}(v') \Rightarrow v \leftarrow v'$ \\
  \> $\mid \mi{app}\,\langle e_1,e_2\rangle \Rightarrow $ \= $v_1 \Leftarrow \mi{eval} \larrow e_1 \semi$ \\
  \>\> $v_2 \Leftarrow^N \mi{eval} \larrow e_2 \semi$ \\
  \>\> $\mb{case}\; v_1$ $(\, \m{lam}(f) \Rightarrow $ \= $e_1' \Leftarrow f.\langle v_2,e_1'\rangle \semi$ \\
  \>\>\> $v \Leftarrow \mi{eval} \larrow e_1'\, )\, )$
\end{tabbing}

\section{Functions}
\label{sec:functions}

Rather than presenting an embedding or translation of a full (sequential) functional
language into our system, we will focus on the case of functions. There is a standard
translation of natural deduction to sequent calculus taking introduction rules to right
rules, and constructing elimination rules from cut and left rules. We base our embedding
of functions into our language on this translation. By following a similar process with
other types, one can similarly embed other functional constructs, such as products and sums.

We will embed functions into an instance of \langname\ with a single mode $m$.
For this example, we specify $\sigma(m) = \{W, C\}$ in order to model a typical
functional language, but we could, for instance, take $\sigma(m) = \{\}$ to model the
linear $\lambda$-calculus. We also restrict the language at mode $m$ to only have sequential cuts,
which will allow us to better model a sequential language.
Note that while we only specify one mode here, we could work within a larger mode structure, as
long as it contains a suitable mode $m$ at which to implement functions --- namely, one with
the appropriate structural properties, and where we have the restriction of only having sequential
cuts. It is this modularity that allows us to freely combine the various reconstructions presented
here and in the following sections.
As we are only working within a single mode in this section, we will generally omit mode subscripts,
but everything is implicitly at mode $m$.

Now, to add functions to this language, we begin by adding a new type $A \to B$ and two new constructs
--- a constructor and a destructor for this type. The constructor, $z.(\lambda x . P_\star)$, writes a $\lambda$-abstraction to destination $z$. Here, we write $P_\star$ to denote that the process expression $P$ denotes its destination by
$\star$.  We will write $P_y$ for $P[y/\star]$. The use of $\star$ makes this closer to the
functional style, where the location that the result is returned to is not made explicit.
The destructor, $P_\star(Q_\star)$, applies the function $P_\star$ to $Q_\star$. These can be typed using
variants of the standard $\to I$ and $\to E$ rules labeled with channels:

\begin{srules}
\infer[\to I]
  {\Gamma \seq z.(\lambda x . P_\star) :: (z : A \to B)}
  {\Gamma, (x : A) \seq P_\star :: (\star : B)}
\qquad 
\infer[\to E]
  {\Gamma \seq y \leftarrow (P_\star(Q_\star)) :: (y : B)}
  {\Gamma \seq P_\star :: (\star : A \to B)
  &\Gamma \seq Q_\star :: (\star : A)}
\end{srules}

In order to avoid having to augment our language each time we wish
to add a new feature, we will show that these new constructs can
be treated as syntactic sugar for terms already in the language, and,
moreover, that those terms behave as we would expect of functions and
function applications.

We take the following definitions for the new type and terms:
\[
\begin{array}{rcl}
A \to B &\triangleq& A \lolli B \\[1ex]
z.(\lambda x . P_\star) &\triangleq&
\mb{case}\, z\, (\langle x, y \rangle \Rightarrow P_y) \\
y \leftarrow (P_\star(Q_\star)) &\triangleq&
f \Leftarrow P_f ; x \Leftarrow Q_x ; f.\langle x, y \rangle \\
\end{array}
\]
These definitions are type-correct, as shown by the following theorem:
\begin{theorem}
If we expand all new constructs using $\triangleq$,
then the typing rules rules $\to I$ and $\to E$ above are admissible.
\end{theorem}
We can prove this by deriving the typing rules, using \cref{lem:seqcut}
in a few places.

Now that we have established that we can expand this syntactic sugar for functions in a
type-correct manner, we examine the evaluation behavior of these terms. First, we consider
the lambda abstraction $z.(\lambda x . P_\star)$ and its expansion
$\mb{case}\, z\, (\langle x, y \rangle \Rightarrow P_y)$. A lambda abstraction should already
be a value, and so we might expect that it can be written to memory immediately. Indeed, in
the expansion, we immediately write the continuation $(\langle x, y \rangle \Rightarrow P_y)$,
which serves as the analogue for $(\lambda x . P_\star)$. This term thus behaves as expected.

We expect that when applying a function $P_\star$ to an argument $Q_\star$, we first reduce
$P_\star$ to a value, then reduce $Q_\star$ to a value, and then apply the value of $P_\star$
to the value of $Q_\star$, generally by substitution. In the term
$f \Leftarrow P_f ; x \Leftarrow Q_x ; f.\langle x, y \rangle$, we see exactly this behavior.
We first evaluate $P_f$ into $f$, then $Q_x$ into $x$, and then apply the continuation stored
in $f$ to the pair $\langle x, y \rangle$. The addition of $y$ allows us to specify
the destination for the result of the function application, as in the
\emph{destination-passing style}~\cite{Wadler84slfp,Larus89thesis,Cervesato02tr,Simmons12phd}
of semantics for functional languages.

\section{Futures}
\label{sec:futures}

Futures~\cite{Halstead85} are a classic example of a primitive
to introduce concurrency into a sequential language. In the usual
presentation, we add to a (sequential) functional language the
ability to create a \emph{future} that immediately returns a
\emph{promise} and spawns a concurrent computation. \emph{Touching}
a promise by trying to access its value blocks until that value has
been computed.  Futures have been a popular mechanism for parallel
execution in both statically and dynamically typed languages, and
they are also used to encapsulate various communication primitives.

The development of a sequential cut in \cref{sec:sem-seq} provides
us with ways to model or reconstruct concurrency primitives, and
futures are a surprisingly simple example of this. Starting
with a language that only allows sequential cuts, we would like to
add a new construct that serves to create a future, as we added
functions to the base language in \cref{sec:functions}. In this case,
however, we already have a construct that behaves exactly as desired.
The concurrent cut $x \leftarrow P \semi Q$ spawns a new process $P$,
and executes $P$ and $Q$ concurrently. When $Q$ tries to read from $x$,
it will block until $P$ has computed a result $W$ and written it to $x$.
If we wish to add an explicit synchronization point, we can do so with
minimal overhead by making use of identity to read from $x$. For instance,
the process $z \Leftarrow (z \leftarrow x) \semi Q$ will first copy or
move the contents of cell $x$ to cell $z$, and then run $Q$. As such, it
delays the execution of $Q$ until $x$ has been written to, even if $Q$ does
not need to look at the value of $x$ until later. This is analogous to
the touch construct of some approaches to futures.

In other words, in this language, futures, rather than being a construct
that we need to add and examine carefully, are in fact the default. This
is, in a sense, opposite to the standard approach, where sequentiality is
the norm and a new construct is needed to handle futures. By instead adding
sequential cut to our otherwise concurrent language, we get the same
expressive power, being able to specify whenever we spawn a new computation
whether it should be run concurrently with or sequentially before the
continuation process.

This approach to futures, much like those in Halstead's Multilisp, does
not distinguish futures at the type level and does not require an explicit
touch construct for synchronization, although we can add synchronization points
as shown. It is possible to provide an encoding
of futures with a distinct type, as they are used in many more modern
languages (see \cref{app:futures-sync}), but we find the form presented
here more natural, as it allows a great deal of flexibility to the programmer ---
a process using a variable $x$ does not know and need not care whether the value of
$x$ is computed concurrently or not.

One interesting result that arises from this approach to futures, and
in particular from the fact that this approach works at any mode $m$,
regardless of what $\sigma(m)$ is, is that by considering the case where
$\sigma(m) = \{\}$, we recover a definition of \emph{linear futures}, which
must be used exactly once. This is limited in that the base language at
mode $m$ will also be linear, along with its futures. However, we are not
restricted to working with one mode. For instance, we may take a mode
$\mS$ with $\sigma(\mS) = \{\}$, which allows for programming linearly with
futures, and a mode $\mS^*$ with $\sigma(\mS^*) = \{\m{W}, \m{C}\}$ and
$\mS < \mS^*$, which allows for standard functional programming.
The shifts between the linear and non-linear modes allow both types of
futures to be used in the same program, embedding the linear language
(including its futures) into the non-linear language via the monad
$\up_\mS^{\mS^*} \down^{\mS^*}_\mS$. Uses for linear futures (without
a full formalization) in the efficient expression of certain parallel
algorithms have already been explored in prior work~\cite{Blelloch99tocs},
but to our knowledge, no formalization of linear futures has yet been given.

\paragraph{Binary Numbers Revisited.}
The program for $\mi{plus2}$ presented in \cref{sec:examples} is a classic example of a
(rather short-lived) pipeline set up with futures.  For this to exhibit the expected
parallelism, the individual $\mi{succ}$ process should also be concurrent in its recursive
call.
\begin{tabbing}
  $(z : \mi{bin}) \vdash \mi{plus2} :: (x : \mi{bin})$ \\[1ex]
  $x \leftarrow \mi{plus2} \larrow z =$ \\
  \qquad $y \leftarrow \mi{succ}\; z \semi$ \\
  \qquad $x \leftarrow \mi{succ}\; y$
\end{tabbing}
Simple variations (for example, setting up a Boolean circuit on bit streams) follow the
same pattern of composition using futures.

\paragraph{\mi{mapReduce} Revisited.}
As a use of futures, consider making all cuts in $\mi{mapReduce}$ sequential except those
representing a recursive call:
\begin{tabbing}
  $(z : \up^s_m B)\, (f : \up^u_m (B \tensor A \tensor B \lolli B))\, (t : \mi{tree}_A)
  \vdash \mi{mapReduce}_{AB} :: (s : B)$ \\[1ex]
  $s \leftarrow \mi{mapReduce}_{AB} \larrow z\; f\; t =$ \\
  \qquad $\mb{case}\; t\; $ \= $(\, \m{empty}\,\langle\,\rangle \Rightarrow $ \= $ z' \Leftarrow z.\mb{shift}(z')$
  \hspace{7em} \= \% \it drop $f$ \\
  \>\> $s \leftarrow z'$ \\
  \> $\mid \m{node}\,\langle l,\langle x,r\rangle\rangle \Rightarrow $ \= $l' \leftarrow \mi{mapReduce}_{AB}\; z\; f\; l \semi$ \\
  \>\> $r' \leftarrow \mi{mapReduce}_{AB}\; z\; f\; r \semi$ \> \% \it duplicate $z$ and $f$ \\
  \>\> $p \Leftarrow p.\langle l',\langle x,r'\rangle\rangle \semi$ \\
  \>\> $s' \Leftarrow f.\mb{shift}\langle p,s'\rangle \semi$ \\
  \>\> $s \leftarrow s'\, )$
\end{tabbing}
In this program, the computation at each node is sequential, but the two recursive calls
to $\mi{mapReduce}$ are spawned as futures.  We synchronize on these futures when
they are needed in the computation of $f$.

\section{Fork/Join Parallelism}
\label{sec:fork-join}

While futures allow us a great deal of freedom in writing concurrent
programs with fine-grained control, sometimes it is useful to have a
more restrictive concurrency primitive, either for implementation reasons
or for reasoning about the program. 
Fork/join parallelism is a simple, yet practically highly successful
paradigm, allowing multiple independent threads to run in parallel,
and then collecting the results together after those threads are finished,
using a \emph{join} construct.
Many slightly different treatments of fork/join exist. Here, 
we will take as the primitive construct a parallel pair
$\langle P_\star \mid Q_\star \rangle$, which runs $P_\star$ and $Q_\star$ in parallel,
and then stores the pair of results. Joining the computation
then occurs when the pair is read from, which requires both $P_\star$ and
$Q_\star$ to have terminated. As with our reconstruction of functions in
\cref{sec:functions}, we will use a single mode $m$ which may have
arbitrary structural properties, but only allows sequential cuts.
As we are working with only a single mode, we will generally omit
the subscripts that indicate mode, writing $A$ rather than $A_m$.

We introduce a new type $A_m \Vert B_m$
of parallel pairs and new terms to create and read from such pairs. We
present these terms in the following typing rules:
\begin{srules}
\infer[\Vert R]
  {\Gamma_C, \Gamma, \Delta \vdash z.\langle P_\star \mid Q_\star \rangle :: (z : A \Vert B)}
  {\Gamma_C, \Gamma \vdash P_\star :: (\star : A)
  &\Gamma_C, \Delta \vdash Q_\star :: (\star : B)}
\\[1em]
\infer[\Vert L]
  {\Gamma, (x : A \Vert B) \vdash \mb{case}\; x\; (\langle z \mid w \rangle \Rightarrow R) :: (c : C)}
  {\Gamma, (z : A), (w : B) \vdash R :: (c : C)}
\end{srules}

As in \cref{sec:functions} we can reconstruct
these types and terms in \langname\ already. Here, we define:
\[
\begin{array}{rcl}
A \Vert B &\triangleq&
A \otimes B \\[1ex]
z.\langle P_\star \mid Q_\star \rangle &\triangleq&
x' \leftarrow P_\star[x'.\mb{shift}(x) /\!\!/ \star] \semi \\
&& y' \leftarrow Q_\star[y'.\mb{shift}(y) /\!\!/ \star] \semi \\
&& \mb{case}\; x'\; (\mb{shift}(x) \Rightarrow \mb{case}\; y'\; (\mb{shift}(y) \Rightarrow z.\langle x, y \rangle )) \\
\mb{case}\; x\; (\langle z \mid w \rangle \Rightarrow R) &\triangleq&
\mb{case}\; x\; (\langle z, w \rangle \Rightarrow R)\\
\end{array}
\]

This definition respects the typing as prescribed by the $\Vert R$
and $\Vert L$ rules.
\begin{theorem}
If we expand all new constructs using $\triangleq$,
then the $\Vert R$ and $\Vert L$ rules above are admissible.
\end{theorem}
This theorem follows quite straightforwardly from \cref{lem:seq-shift-typing}.

The evaluation behavior of these parallel pairs is quite simple --- we first
observe that, as the derivation of $\Vert L$ in the theorem above suggests,
the reader of a parallel pair behaves exactly as the reader of an ordinary
pair. The only difference, then, is in the synchronization behavior of the
writer of the pair. Examining the term
\[
\begin{array}{l}
x' \leftarrow P_\star[x'.\mb{shift}(x) /\!\!/ \star] \semi \\
y' \leftarrow Q_\star[y'.\mb{shift}(y) /\!\!/ \star] \semi \\
\mb{case}\; x'\; (\mb{shift}(x) \Rightarrow \mb{case}\; y'\; (\mb{shift}(y) \Rightarrow z.\langle x, y \rangle )) \\
\end{array}
\]
we see that it spawns two new threads, which run concurrently with
the original thread.
The new threads execute $P_\star[x'.\mb{shift}(x) /\!\!/ \star]$ and
$Q_\star[y'.\mb{shift}(y) /\!\!/ \star]$ with destinations $x'$ and $y'$,
respectively, while the original thread waits first on $x'$, then
on $y'$, before writing the pair $\langle x, y \rangle$ to $z$.
Because the new threads will write to $x$ and $x'$ atomically, and
similarly for $y$ and $y'$, by the time $\langle x, y \rangle$ is
written to $z$, $x$ and $y$ must have already been written to.
However, because both cuts in this term are concurrent cuts, $P_\star$
and $Q_\star$ run concurrently, as we expect from a parallel pair.

\paragraph{$\mi{mapReduce}$ Revisited.}
We can use the fork/join pattern in the implementation of $\mi{mapReduce}$
so that we first synchronize on the results returned from the two
recursive calls before we call $f$ on them.
\begin{tabbing}
  $(z : \up^s_m B)\, (f : \up^u_m (B \tensor A \tensor B \lolli B))\, (t : \mi{tree}_A)
  \vdash \mi{mapReduce}_{AB} :: (s : B)$ \\[1ex]
  $s \leftarrow \mi{mapReduce}_{AB} \larrow z\; f\; t =$ \\
  \qquad $\mb{case}\; t\; $ \= $(\, \m{empty}\,\langle\,\rangle \Rightarrow $ \= $ z' \Leftarrow z.\mb{shift}(z')$ \\
  \>\> $s \leftarrow z'$ \\
  \> $\mid \m{node}\,\langle l,\langle x,r\rangle\rangle \Rightarrow $ \\
  \> \qquad \= $rl \leftarrow rl.\langle \mi{mapReduce}_{AB}\; x\; f\; l \mid \mi{mapReduce}_{AB}\; x\; f\; r\rangle \semi$ \\
  \>\> $\mb{case}\; rl\; (\langle l' \mid r' \rangle \Rightarrow$
  \= $p \Leftarrow p.\langle l',\langle x,r'\rangle\rangle \semi$ \\
  \>\>\> $s' \Leftarrow f.\mb{shift}\langle p,s'\rangle \semi$ \\
  \>\>\> $s \leftarrow s'\,) \,)$
\end{tabbing}

\section{Monadic Concurrency}
\label{sec:sill}

For a different type of concurrency primitive, we look at a
monad for concurrency, taking some inspiration from
SILL~\cite{Toninho13esop,Toninho15phd,Griffith16phd}, which
makes use of a contextual monad to embed the concurrency
primitives of linear session types into a functional language. This
allows us to have both a fully-featured sequential functional
language and a fully-featured concurrent linear language, with the
concurrent layer able to refer on variables in the sequential
layer, but not the other way around.

To construct this concurrency monad, we will use two modes $\mN$ and $\mS$
with $\mN < \mS$. Intuitively, the linear concurrent portion of the
language is at mode $\mN$, while the functional portion is at mode
$\mS$. As in common in functional languages, $\mS$ allows weakening
and contraction ($\sigma(\mS) = \{W,C\}$), but only permits sequential
cuts (by which we mean that any cut whose principal formula is at mode
$\mS$ must be a sequential cut) so that it models a sequential
functional language. By contrast, $\mN$ allows concurrent cuts, but is
linear ($\sigma(\mS) = \{\}$). We will write $A_\mS$ and $A_\mN$ for
sequential and concurrent types, respectively.

We will borrow notation from SILL, using the type $\{ A_\mN \}$
for the monad, and types $A_\mS \land B_\mN$ and $A_\mS \supset B_\mN$
to send and receive functional values in the concurrent layer, respectively.
The type $\{ A_\mN \}$ has as values process expressions
$\{ P_\star\}$ such that ${P_\star :: (\star : A_\mN)}$.
These process expressions can be constructed and passed around in
the functional layer. In order to actually execute these processes,
however, we need to use a \emph{bind} construct
$\{c_\mN\} \leftarrow Q_\star$ in the functional layer, which will
evaluate $Q_\star$ into an encapsulated process expression $\{P_\star\}$
and then run $P_\star$, storing its result in $c_\mN$.
We can add $\{\cdot\}$ to our language with the typing rules below.
Here, $\Gamma_\mS$ indicates that all assumptions in $\Gamma$ are at mode $\mS$:
\begin{srules}
\infer[\{\cdot\}I]
  {\Gamma_\mS \vdash y_\mS.\{P_\star\} :: (y_\mS :: \{A_\mN\})}
  {\Gamma_\mS \vdash P_\star :: (\star_\mN : A_\mN)}
\qquad
\qquad
\infer[\{\cdot\}E]
  {\Gamma_\mS \vdash \{c_\mN\} \leftarrow Q_\star :: (c_\mN : A_\mN)}
  {\Gamma_\mS \vdash Q_\star :: (\star_\mS :: \{A_\mN\})}
\end{srules}

Since they live in the session-typed layer, the $\land$ and $\supset$
constructs fit more straightforwardly into our language. We will focus on the type
$A_\mS \land B_\mN$, but $A_\mS \supset B_\mN$ can be handled similarly. A process of type
$A_\mS \land B_\mN$ should write a pair of a functional value with type $A_\mS$ and a
concurrent value with type $B_\mN$. These terms and their typing rules are shown below:
\begin{srules}
\infer[\land R^0]
  {\Gamma_W, (v_\mS : A_\mS), (y_\mN : B_\mN) \vdash
   x_\mN.\langle v_\mS, y_\mN \rangle :: (x_\mN :: A_\mS \land B_\mN)
  }
  {\mathstrut}
\end{srules}
\begin{srules}
\infer[\land L]
  {\Gamma, (x_\mN : A_\mS \land B_\mN) \vdash
   \mb{case}\; x_\mN\; (\langle v_\mS, y_\mN \rangle \Rightarrow P_z) :: (z_\mN : C_\mN)
  }
  {\Gamma, (v_\mS : A_\mS), (y_\mN : A_\mN) \vdash P_z :: (z_\mN : C_\mN)
  }
\end{srules}

To bring these new constructs into the base language, we define
\[
\begin{array}{rcl}
\{A_\mN\} &\triangleq& \up_\mN^\mS A_\mN \\
A_\mS \land B_\mN &\triangleq& \left(\down^\mS_\mN A_\mS\right) \otimes B_\mN \\[1ex]
d_\mS.\{P_\star\} &\triangleq&
\mb{case}\; d_\mS\; (\mb{shift}(x_\mN) \Rightarrow P_x) \\
\{c_\mN\} \leftarrow Q_\star &\triangleq&
y_\mS \Leftarrow Q_y ; y_\mS.\mb{shift}(c_\mN) \\[1ex]
d_\mN.\langle v_\mS, y_\mN \rangle &\triangleq&
x_\mN \leftarrow x_\mN.\mb{shift}(v_\mS) ; d_\mN.\langle x_\mN, y_\mN \rangle \\
\mb{case}\; d_\mN\; (\langle u_\mS, w_\mN \rangle \Rightarrow P_z) &\triangleq&
\mb{case}\; d_\mN\; (\langle x_\mN, w_\mN \rangle \Rightarrow \mb{case}\; x_\mN (\mb{shift}(v_\mS) \Rightarrow P_z)) \\
\end{array}
\]

These definitions give us the usual type-correctness theorem:
\begin{theorem}
If we expand all new constructs using $\triangleq$,
then the typing rules for $\{\cdot\}$ and $\land$ are admissible.
\end{theorem}

As with the previous sections, it is not enough to know that these definitions
are well-typed --- we would also like to verify that they have the behavior
we expect for a concurrency monad. In both cases,
this is relatively straightforward. Examining the term
\[
d_\mS.\{P_\star\} \quad \triangleq \quad
\mb{case}\; d_\mS\; (\mb{shift}(x_\mN) \Rightarrow P_x),
\]
we see that this writes a continuation into memory, containing the
process $P_x$. A reference to this continuation can then be passed around
freely, until it is executed using the bind construct:
\[
\{c_\mN\} \leftarrow P_\star \quad \triangleq \quad
y_\mS \Leftarrow P_y ; y_\mS.\mb{shift}(c_\mN) 
\]
This construct first evaluates $P_y$ with destination $y_\mS$, to get a stored
process, and then executes that stored process with destination $c_\mN$.

The $\land$ construct is even simpler. Writing a functional value using the term
\[
d_\mN.\langle v_\mS, y_\mN \rangle
\quad \triangleq \quad
x_\mN \leftarrow x_\mN.\mb{shift}(v_\mS) ; d_\mN.\langle x_\mN, y_\mN \rangle
\]
sends both a shift (bringing the functional value into the concurrent
layer) and the pair $\langle x_\mN, y_\mN \rangle$ of the continuation $y_\mN$
and the shift-encapsulated value $x_\mN$. Reading such a value using the term
\[
\mb{case}\; d_\mN\; (\langle v_\mS, y_\mN \rangle \Rightarrow P_z) 
\quad \triangleq \quad
\mb{case}\; d_\mN\; (\langle x_\mN, y_\mN \rangle \Rightarrow \mb{case}\; x_\mN (\mb{shift}(v_\mS) \Rightarrow P_z))
\]
just does the opposite --- we read the pair out of memory, peel the shift off of
the functional value $v_\mS$ to return it to the sequential, functional layer, and
continue with the process $P_z$, which may make use of both $v_\mS$ and the continuation
$y_\mN$.

These terms therefore capture the general behavior of a monad used to encapsulate concurrency
inside a functional language. The details of the monad we present here are different from
that of SILL's (contextual) monad, despite our use of similar notation, but the essential idea
is the same.

\paragraph{Example: A Concurrent Counter.}
We continue our example of binary numbers, this time supposing that the mode $m = \mS$,
that is, our numbers and the successor function on them are sequential and allow weakening
and contraction.  $\mi{counter}$ represents a concurrently running process that can
receive $\m{inc}$ and $\m{val}$ messages to increment or retrieve the counter value,
respectively.
\begin{tabbing}
$\mi{ctr}_\mN = \with_\mN \{\m{inc} : \mi{ctr}_\mN, \m{val} : \mi{bin}_\mS \land \mi{ctr}_\mN\}$ \\[1ex]
$(x : \mi{bin}_\mS) \vdash \mi{counter} :: (c : \mi{ctr}_\mN)$ \\[1ex]
$c \leftarrow \mi{counter} \larrow x =$ \\
\qquad $\mb{case}\; c\; $ \= $(\, \m{inc}(c') \Rightarrow$ \= $ x' \Leftarrow \mi{succ}\; x \semi$ \\
\>\> $c' \leftarrow \mi{counter}\; x'$ \\
\> $|\, \m{val}(c') \Rightarrow$ \= $c'' \leftarrow \mi{counter}\; x \semi$ \\
\>\> $c'.\langle x, c'' \rangle$
\end{tabbing}

\section{Conclusion}


We have presented a concurrent shared-memory semantics based on a
semi-axiomatic~\cite{DeYoung20fscd} presentation of adjoint
logic~\cite{Reed09un,Licata16lfcs,Licata17fscd,Pruiksma19places}, for which
we have usual variants of progress and preservation, as well as confluence.
We then demonstrate that by adding a limited form of atomic writes, we can
model \emph{sequential} computation. Taking advantage of this, we reconstruct
several patterns that provide limited access to concurrency in
a sequential language, such as fork/join, futures, and monadic concurrency in
the style of SILL. The uniform nature of these reconstructions means that they
are all mutually compatible, and so we can freely work with any set of these
concurrency primitives within the same language.

There are several potential directions that future work in this space could take.
In our reconstruction of futures, we incidentally also provide a definition of
\emph{linear} futures, which have been used in designing pipelines~\cite{Blelloch99tocs},
but to our knowledge have not been examined formally or implemented. One item of
future work, then, would be to further explore linear futures, now aided by a
formal definition which is also amenable to implementation. We also believe that
it would be interesting to explore an implementation of our language as a whole,
and to investigate what other concurrency patterns arise naturally when working
in it.
Additionally, the stratification of the language into layers connected with adjoint
operators strongly suggests that some properties of a language instance as a whole can
be obtained modularly from properties of the sublanguages at each mode. Although based on
different primitives, research on monads and comonads to capture effects and coeffects,
respectively,~\cite{Curien16popl,Gaboardi16icfp} also points in this direction. In
particular, we would like to explore a modular theory of (observational) equivalence
using this approach. Some work on observational equivalence in a substructural setting
already exists,~\cite{Kavanagh20express} but works in a message-passing setting and does not
seem to translate directly to the shared-memory setting of \langname.

\bibliographystyle{splncs04}
\bibliography{fp}

\appendix
\clearpage

\section{Typed Futures}
\label{app:futures-sync}

The futures that we discuss in \cref{sec:futures} behave much like
Halstead's original futures in Multilisp, which, rather than being
distinguished at the type level, are purely operational. One side
effect of this is that while we can explicitly synchronize these
futures, we can also make use of implicit synchronization, where
accessing the value of the future blocks until it has been computed,
without the need for a touch construct.

Here, we will look at a different encoding of futures, which distinguishes
futures at the type level, as they have often been presented since.
As in \cref{sec:functions}, we will work with a single mode $m$,
in which we will only allow sequential cuts, and which may have any
set $\sigma(m)$ of structural properties. To the base language, we
add the following new types and process terms for futures:
\[
  \begin{array}{llcl}
    \mbox{Types} & A & ::= & \ldots \mid \mb{fut}\, A \\
    \mbox{Processes} & P & ::= & \ldots \mid x.\langle P_\star \rangle \mid \mb{touch}\, y\, (\langle z\rangle \Rightarrow P)
  \end{array}
\]
We type these new constructs as:
\begin{rules}
  \infer[{\mb{fut}}R]
  {\Gamma \vdash x_m.\langle P_\star \rangle :: (x_m : \mb{fut}\; A_m)}
  {\Gamma \vdash P_\star :: (\star : A_m)}
  \\[1em]
  \infer[{\mb{fut}}L]
  {\Gamma, x_m : \mb{fut}\, A_m \vdash \mb{touch}\, x_m\, (\langle z_m\rangle \Rightarrow Q) :: (w_m : C_m)}
  {\Gamma, z_m : A_m \vdash Q :: (w_m : C_m)}
\end{rules}%

We then reconstruct this in {\langname} by defining
\[
  \begin{array}{lcl}
    \mb{fut}\, A_m & \triangleq & \down^m_m \down^m_m A_m \\[1ex]
    x_m.\langle P_\star \rangle
    & \triangleq &
    z_m' \leftarrow P_\star[z_m'.\mb{shift}(z_m) /\!\!/ \star] \semi x_m.\mb{shift}(z_m') \\
    \mb{touch}\, x_m\, (\langle z_m\rangle \Rightarrow Q)
    & \triangleq &
    \mb{case}\, x_m\, (\mb{shift}(z_m') \Rightarrow \mb{case}\, z_m' (\mb{shift}(z_m) \Rightarrow Q)) \\
  \end{array}
\]
This is not the only possible reconstruction,
\footnote{
In particular, as the role of the outer shift is simply to allow the
client of the future to proceed, we can replace the shift with any other
type that forces a send but does not provide any useful information.
Examples include $\up_m^m \down^m_m A_m$ and $\one_m \otimes (\down^m_m A_m)$.
}
but we use it because it is the simplest one that we have found.
The first property to verify is that these definitions are type-correct:
\begin{theorem}
If we expand all new constructs using $\triangleq$,
then the rules ${\mb{fut} L}$ and ${\mb{fut} R}$ are admissible.
\end{theorem}
\begin{proof}
By examining typing derivations for these processes, we see that
these rules can be derived as follows:
\begin{small}
\[
\infer[\cut]
  {\Gamma \vdash z_m' \leftarrow P_\star[z_m'.\mb{shift}(z_m) /\!\!/ \star] \semi
                 x_m.\mb{shift}(z_m') :: (x_m : \down^m_m \down^m_m A_m)
  }
  {\infer-[\cref{lem:seq-shift-typing}]
    {\Gamma \vdash P_\star[z_m'.\mb{shift}(z_m) /\!\!/ \star] :: (z_m' : \down^m_m A_m)}
    {\Gamma \vdash P_\star :: (\star :: A_m)}
  &\infer[\down R^0]
    {z_m' : \down^m_m A_m \vdash x_m.\mb{shift}(z_m') :: (x_m : \down^m_m \down^m_m A_m)}
    {\mathstrut}
  }
\]
\vspace{1.5em}
\[
\infer[\down L_0]
  {\Gamma, x_m : \down^m_m \down^m_m A_m \vdash
   \mb{case}\, x_m\, (\mb{shift}(z_m') \Rightarrow
   \mb{case}\, z_m' (\mb{shift}(z_m) \Rightarrow Q)) :: (w : C_m)}
  {\infer[\down L_0]
     {\Gamma, z_m' : \down^m_m A_m \vdash
      \mb{case}\, z_m' (\mb{shift}(z_m) \Rightarrow Q) :: (w : C_m)}
     {\Gamma, z_m : A_m \vdash Q :: (w : C_m)}
  }
\]
\end{small}
Note that we omit mode conditions on cut because within a single mode $m$,
they are necessarily satisfied.
\end{proof}

Now, we examine the computational behavior of these terms to demonstrate that
they behave as futures. The type $\down^m_m A_m$, much like in \cref{sec:sem-seq}
where we used it to model sequentiality, adds an extra synchronization point.
Here, we shift twice, giving $\down^m_m \down^m_m A_m$, to introduce two
synchronization points. The first is that enforced by our restriction to only
allow sequential cuts in this language (outside of futures), while the second
will become the $\mb{touch}$ construct. We will see both of these when we examine
each process term.

We begin by examining the constructor for futures. Intuitively, when creating a future,
we would like to spawn a new thread
to evaluate $P_\star$ with new destination $z_m$, and immediately write the promise of $z_m$
(represented by a hypothetical new value $\langle z_m \rangle$) into $x_m$, so that
any process waiting on $x_m$ can immediately proceed. The term
\[
x_m.\langle P_\star \rangle
\quad \triangleq \quad 
z_m' \Leftarrow P_\star[z_m'.\mb{shift}(z_m) /\!\!/ \star] \semi x_m.\mb{shift}(z_m')
\]
behaves almost exactly as expected. Rather than spawning $P_\star$ with destination $z_m$,
we spawn $P_\star[z_m'.\mb{shift}(z_m) /\!\!/ \star]$, which will write the result of
$P_\star$ to $z_m$, and a synchronizing shift to $z_m'$. Concurrently, we write the value
$\mb{shift}(z_m')$ to $x_m$, allowing the client of $x_m$ to resume execution, even if
$x_m$ was created by a sequential cut. This value $\mb{shift}(z_m')$ is the first half
of the promise $\langle z_m \rangle$, and the second half, $\mb{shift}(z_m)$, will be
written to $z_m'$ when $P$ finishes executing.

If, while $P$ continues to execute, we touch $x_m$, we would expect to block until
the promise $\langle z_m \rangle$ has been fulfilled by $P$ having written to $z_m$.
Again, we see exactly this behavior from the term
\[
\mb{touch}\, x_m\, (\langle z_m\rangle \Rightarrow Q)
\quad \triangleq \quad
\mb{case}\, x_m\, (\mb{shift}(z_m') \Rightarrow \mb{case}\, z_m' (\mb{shift}(z_m) \Rightarrow Q)).
\]
This process will successfully read $\mb{shift}(z_m')$ from $x_m$, but will block trying to
read from $z_m'$ until $z_m'$ is written to. Since $z_m$ and $z_m'$ are written to at the
same time, we block until $z_m$ is written to, at which point the promise is fulfilled.
Once a result $W$ has been written to $z_m$ and (simultaneously) $\mb{shift}(z_m)$ has
been written to $z_m'$, this process can continue, reading both $z_m'$ and $z_m$, and
continuing as $Q$. Again, this is the behavior we expect a touch construct to have.

This approach does effectively model a form of typed future, which ensures that all
synchronization is explicit, but comes at the cost of overhead from the additional
shifts. Both this and the simpler futures that we describe in \cref{sec:futures}
have their uses, but we believe that the futures native to {\langname} are more
intuitive in general.

\section{Proofs of Type Correctness}
\label{app:trans-type-proofs}

In \cref{sec:functions,sec:fork-join,sec:sill}, we present
type-correctness theorems for our reconstructions of various
concurrency primitives, but omit the details of the proofs.
Here, we present those details.

\paragraph{Functions.}
We derive the typing rules as follows, making use of \cref{lem:seqcut}
to use the admissible $\m{seqcut}$ rule. We omit the conditions on modes
for cut, as we only have one mode:
\[
\infer[\lolli R]
  {\Gamma \seq \mb{case}\, z\, (\langle x, y \rangle \Rightarrow P_y) :: (z : A \lolli B)}
  {\Gamma, x : A \seq P_y :: (y : B)}
\]
\[
\infer-[\m{seqcut}]
  {\Gamma \seq f \Leftarrow P_f ; x \Leftarrow Q_x ; f.\langle x, y \rangle :: (y : B)}
  {\Gamma \seq P_f :: (f : A \lolli B)
  &\infer-[\m{seqcut}]
     {\Gamma, (f : A \lolli B) \seq x \Leftarrow Q_x ; f.\langle x, y \rangle :: (y : B)}
     {\Gamma \seq Q_x :: (x : A)
     &\infer[\lolli L^0]
        {\Gamma, (f : A \lolli B), (x : A) \seq f.\langle x, y \rangle :: (y : B)}
        {\mathstrut}
     }
  }
\]

\paragraph{Fork/Join.}
Due to the length of the process term that defines
$z.\langle P_\star \mid Q_\star \rangle$, we elide portions of it throughout the derivation
below, and we will write $P'$ for $P_\star[x'.\mb{shift}(x) /\!\!/ \star]$, and similarly
$Q'$ for $Q_\star[y'.\mb{shift}(y) /\!\!/ \star]$. With these abbreviations, we have
the following derivation for the $\Vert R$ rule, where the dashed inferences
are made via \cref{lem:seq-shift-typing}.
\begin{srules}
\infer[\cut]
  {\Gamma_C, \Gamma, \Delta \vdash x' \leftarrow P' \semi \ldots :: (z : A \otimes B)}
  {\infer-
     {\Gamma_C, \Gamma \vdash P' :: (x' : \down^m_m A)}
     {\Gamma_C, \Gamma \vdash P_x :: (x : A)}
  &\infer[\cut]
     {\Gamma_C, \Delta, x' : \down^m_m A \vdash y' \leftarrow Q' \semi \ldots :: (z : A \otimes B)}
     {\infer-
       {\Gamma_C, \Delta \vdash Q' :: (y' : \down^m_m B)}
       {\Gamma_C, \Delta \vdash Q_y :: (y : B)}
     &\infer[\down L_0]
        {x' : \down^m_m A, y' : \down^m_m B \vdash
         \mb{case}\; x'\; (\ldots) :: (z : A \otimes B)}
        {\infer[\down L_0]
           {x : A, y' : \down^m_m B \vdash
            \mb{case}\; y'\; (\ldots) :: (z : A \otimes B)}
           {\infer[\otimes R^0]
              {x : A, y : B \vdash z.\langle x, y \rangle :: (z : A \otimes B)}
              {\mathstrut}
           }
        }
     }
  }
\end{srules}
The left rule is much more straightforward, since this encoding
makes the writer of the pair rather than the reader responsible for synchronization.
\begin{srules}
\infer[\otimes L_0]
  {\Gamma, x : A \otimes B \vdash \mb{case}\; x\; (\langle z, w \rangle \Rightarrow R) :: (c : C)}
  {\Gamma, z: A, w : B \vdash R :: (c : C)
  }
\end{srules}

\paragraph{Monadic Concurrency.}
We first construct the typing rules for $\{\cdot\}$, which are straightforward:
\begin{srules}
\infer[\up R]
  {\Gamma_\mS \seq \mb{case}\; d_\mS\; (\mb{shift}(x_\mN) \Rightarrow P_x) :: (d_\mS : \up_\mN^\mS A_\mN)}
  {\Gamma_\mS \seq P_x :: (x_\mN : A_\mN)}
\end{srules}
\begin{srules}
\infer-[\m{seqcut}]
  {\Gamma_\mS \seq y_\mS \Leftarrow P_y ; y_\mS.\mb{shift}(c_\mN) :: (c_\mN : A_\mN)}
  {\Gamma_\mS \geq \mS \geq \mN
  &\Gamma_\mS \seq P_y :: (y_\mS : \up_\mN^\mS A_\mN)
  &\infer[\up L^0]
     {(y_\mS : \up_\mN^\mS A_\mN) \seq y_\mS.\mb{shift}(c_\mN) :: (c_\mN : A_\mN)}
     {\mathstrut}
  }
\end{srules}
We then construct the typing rules for $\land$:
\begin{srules}
\infer[\cut]
  {\Gamma_W, (v_\mS : A_\mS), (y_\mN : B_\mN) \vdash
   x_\mN \leftarrow x_\mN.\mb{shift}(v_\mS) ; d_\mN.\langle x_\mN, y_\mN \rangle :: (d_\mN :: A_\mS \land B_\mN)
  }
  {
  &\infer[\down R^0]
     {v_\mS : A_\mS \vdash x_\mN.\mb{shift}(v_\mS) :: (x_\mN :: \down^\mS_\mN A_\mS)}
     {\mathstrut}
  &\infer[\otimes R^0]
     {\Gamma_W, (x_\mN :: \down^\mS_\mN A_\mS) \vdash d_\mN.\langle x_\mN, y_\mN \rangle :: (d_\mN :: (\down^\mS_\mN A_\mS) \otimes B_\mN)}
     {\mathstrut}
  }
\end{srules}
\begin{srules}
\infer[\otimes L_0]
  {\Gamma, (d_\mN : (\down^\mS_\mN A_\mS) \otimes B_\mN) \vdash
   \mb{case}\; d_\mN\; (\langle x_\mN, w_\mN \rangle \Rightarrow \mb{case}\; x_\mN (\mb{shift}(v_\mS) \Rightarrow P_z))
   :: (z_\mN : C_\mN)
  }
  {\infer[\down L_0]
     {\Gamma, (x_\mN : \down^\mS_\mN A_\mS), (w_\mN : B_\mN) \vdash
      \mb{case}\; x_\mN (\mb{shift}(v_\mS) \Rightarrow P_z) :: (z_\mN : C_\mN)
     }
     {\Gamma, (v_\mS : A_\mS), (w_\mN : B_\mN) \vdash P_z :: (z_\mN : C_\mN)}
  }
\end{srules}
Note that unlike the rules for $\{\cdot\}$ or for many of the constructs
in previous sections, those for $\land$ are not only admissible --- they
are \emph{derivable}.

\end{document}